%% file: agatha.tex
\documentclass[sigconf,nonacm]{acmart}[9pt]

\renewcommand\footnotetextcopyrightpermission[1]{} %
\AtBeginDocument{%
  \providecommand\BibTeX{{%
    \normalfont B\kern-0.5em{\scshape i\kern-0.25em b}\kern-0.8em\TeX}}}

\usepackage{tikz}
\usepackage[
	n,
	operators,
	advantage,
	sets,
	adversary,
	landau,
	probability,
	notions,
	logic,
	ff,
	mm,
	primitives,
	events,
	complexity,
	asymptotics,
	keys]{cryptocode}
\usepackage{amsmath}
\usepackage[]{cryptocode}
\usepackage{amsmath,amsfonts}
 \usepackage{algorithm}
\newcommand*\circled[1]{\tikz[baseline=(char.base)]{
             \node[shape=circle,draw,inner sep=0pt,fill=black, text=white] (char) {#1};}}
\usepackage{graphicx}
\usepackage{textcomp}
\usepackage{xcolor}
\usepackage{multirow}
\usepackage{hyperref}
\usepackage{color}
\usepackage{booktabs}

\usepackage{wasysym}
\usepackage[flushleft]{threeparttable}
\usepackage{tablefootnote}
\usepackage{footnote}
\usepackage{subcaption}
 \usepackage{caption}
\newtheorem{theorem}{\bf Theorem}[section]
\newtheorem{lemma}[theorem]{\bf Lemma}

\newenvironment{definition}[1][Definition]{\begin{trivlist}
\item[\hskip \labelsep {\bfseries #1}]}{\end{trivlist}}

\newcommand{\para}[1]{\vspace{2pt}\noindent\textbf{#1}}
\usepackage{titlesec}
\titlespacing{\section}{0pt}{6.5pt}{6.5pt}
\titlespacing{\subsection}{0pt}{5pt}{5pt}

\usepackage{mathrsfs}
\newcounter{Lcount}
\newcommand{\squishlisttwo}{
\begin{list}{\arabic{Lcount}. }
{ \usecounter{Lcount}
\setlength{\itemsep}{0pt}
\setlength{\parsep}{0pt}
\setlength{\topsep}{0pt}
\setlength{\partopsep}{0pt}
\setlength{\leftmargin}{2em}
\setlength{\labelwidth}{1.5em}
\setlength{\labelsep}{0.5em} } }

\newcommand{\squishtwoend}{
\end{list} }

\newcommand{\squishlist}{
   \begin{list}{$\bullet$}
    { \setlength{\itemsep}{0pt}      \setlength{\parsep}{0pt}
      \setlength{\topsep}{3pt}       \setlength{\partopsep}{0pt}
      \setlength{\listparindent}{-2pt}
      \setlength{\itemindent}{-5pt}
      \setlength{\leftmargin}{0.5em} \setlength{\labelwidth}{0em}
      \setlength{\labelsep}{0.5em} } }

\newcommand{\squishend}{
    \end{list}  }

\hypersetup{citecolor=black,linkcolor=black,urlcolor=black}

\begin{document}

\title{Agatha: Smart Contract for DNN Computation}

\author{Zihan Zheng}
\affiliation{%
  \institution{University of Science and Technology of China}
  \city{Hefei}
  \country{China}
}
\email{zzh1996@mail.ustc.edu.cn}

\author{Peichen Xie}
\affiliation{%
  \institution{Peking University}
  \city{Beijing}
  \country{China}
}
\email{xpc@pku.edu.cn}

\author{Xian Zhang}
\affiliation{%
  \institution{Microsoft Research}
  \city{Beijing}
  \country{China}
}
\email{zhxian@microsoft.com}

\author{Shuo Chen}
\affiliation{%
  \institution{Microsoft Research}
  \city{Beijing}
  \country{China}
}
\email{shuochen@microsoft.com}

\author{Yang Chen}
\affiliation{%
  \institution{Microsoft Research}
  \city{Beijing}
  \country{China}
}
\email{yachen@microsoft.com}

\author{Xiaobing Guo}
\affiliation{%
  \institution{Microsoft Research}
  \city{Beijing}
  \country{China}
}
\email{xiaobing.guo@microsoft.com}

\author{Guangzhong Sun}
\affiliation{%
  \institution{University of Science and Technology of China}
  \city{Hefei}
  \country{China}
}
\email{gzsun@ustc.edu.cn}

\author{Guangyu Sun}
\affiliation{%
  \institution{Peking University}
  \city{Beijing}
  \country{China}
}
\email{gsun@pku.edu.cn}

\author{Lidong Zhou}
\affiliation{%
  \institution{Microsoft Research}
  \city{Beijing}
  \country{China}
}
\email{lidongz@microsoft.com}

\renewcommand{\shortauthors}{Zihan Zheng, Peichen Xie, Xian Zhang, Shuo Chen, et al.}

\begin{abstract}

Smart contract is one of the core features of Ethereum and has inspired many blockchain descendants. Since its advent, the verification paradigm of smart contract has been improving toward high scalability. It shifts from the expensive on-chain verification to the orchestration of off-chain VM (virtual machine) execution and on-chain arbitration with the pinpoint protocol. The representative projects are TrueBit, Arbitrum, YODA, ACE, and Optimism. Consequently, verification of more and more complicated computations for smart contract is achieved, such as the aggregated execution of DeFi transactions. Inspired by visionaries in academia and industry, we consider the DNN (deep neural network) computation to be promising but on the next level of complexity for the verification paradigm of smart contract. Unfortunately, even for the state-of-the-art verification paradigm, off-chain VM execution of DNN computation has an orders-of-magnitude slowdown compared to the native off-chain execution.

To enable the native off-chain execution of verifiable DNN computation, we present {\em Agatha} system, which solves the significant challenges of {\em misalignment} and {\em inconsistency}: (1) Native DNN computation has a graph-based computation paradigm misaligned with previous VM-based execution and arbitration; (2) Native DNN computation may be inconsistent cross platforms which invalidates the verification paradigm. In response, we propose the graph-based pinpoint protocol (GPP) which enables the pinpoint protocol on computational graphs, and bridges the native off-chain execution and the contract arbitration. We also develop a technique named Cross-evaluator Consistent Execution (XCE), which guarantees cross-platform consistency and forms the correctness foundation of GPP. We showcase Agatha for the DNN computation of popular models (MobileNet, ResNet50 and VGG16) on Ethereum. Agatha achieves a negligible on-chain overhead, and an off-chain execution overhead of 3.0\%, which represents an off-chain latency reduction of at least 602$\times$ compared to the state-of-the-art verification paradigm.

\end{abstract}

\settopmatter{printfolios=true}
\maketitle

\input{intro}
\input{pre}
\input{overview}
\input{pinpoint}
\input{determine}
\input{security}
\input{eval}
\input{relate}
\bibliographystyle{ACM-Reference-Format}
\bibliography{agatha}
\input{appendix}

\end{document}

%% file: intro.tex
\section{Introduction}\label{sec:intro}

The Turing completeness of {\em smart contract} enables arbitrary computations to be executed and verified on blockchain nodes. It underpins the expressiveness of the decentralized computing paradigm. As a result, more and more complicated applications are emerging on the Ethereum platform \cite{ethereum}, such as ERC-20 tokens \cite{erc20}, Cryptokitties \cite{cryptokitties}, Zero-knowledge Proof verification \cite{zokrates} and Decentralized Finance (DeFi) \cite{uniswap}.

However, current smart contracts are limited by low efficiency for complex applications, since a contract's computation is re-executed and verified on all Ethereum nodes (e.g., miners, full nodes). For example, a simple task of naïve matrix multiplication of 1000$\times$1000 integers would cost over 3 billion gas (the unit of cost in Ethereum), which far exceeds the current Ethereum's block gaslimit (i.e., 15 million). Even if the task could span over 200 blocks to meet the gaslimit, it would result in an extremely low throughput of $4.2\times10^{-4}$ tasks per second, given Ethereum's block interval of $\sim$12 seconds. In comparison, even a low-end dual-core laptop can do this type of simple tasks with a latency of 0.39 seconds and a throughput of 2.5 tasks per second.

To improve contracts' efficiency for complex applications, one of the most promising solutions is to offload a task off-chain to a quorum of nodes and leave only a contract for arbitrating ``fraud proofs'' on-chain. A fraud proof, which is generated via an interactive {\em pinpoint protocol}, proves that the claimed result of a computation is fraudulent/incorrect. The arbitration process is designed to validate a fraud proof by executing only a few minor computation steps, making the on-chain cost extremely low. Existing technologies, such as Plasma \cite{plasma}, TrueBit \cite{truebit}, Arbitrum \cite{arbitrum}, YODA \cite{yoda}, ACE \cite{ace}, Optimism \cite{optimism}, share this key idea. This opens up exciting opportunities for smart contracts to fulfill new scenarios in the next complexity level, such as the escrow contract, iterated hashing, aggregated execution of DeFi transactions, etc.

If the complexity moves up one more level, AI computations naturally become a fascinating paradigm that people want smart contracts to support. Indeed,
researchers in both the academia (e.g., YODA \cite{yoda}, TrueBit \cite{truebit}, ACE \cite{ace}) and the industry (e.g., Microsoft \cite{microsoft} and startups \cite{cortex}, \cite{singularitynet}, \cite{algorithmia}, \cite{oraclize}) have put forth the vision to fulfill AI computations with decentralized consensus. However, this is an uncharted territory, because contemporary AI computations, specifically DNN (deep neural network) computations, impose serious challenges for all existing technologies:
\squishlist
\item {\em Overhead of on-chain execution.} Some proposed technologies \cite{singularitynet,oraclize,algorithmia,microsoft,smarthome,cdda} either target small AI computations or only target the consensus among a small number of nodes. They do AI computations by the on-chain execution, thus cannot scale. If these approaches are used on Ethereum to run a DNN computation, the cost would be prohibitively high. For example, a single VGG16 \cite{vgg} inference is estimated to consume $\sim 90$ billion gas, equivalent to running $4.3\times10^6$ ETH-transfer transactions\footnote{The ETH-transfer transaction (ETH, or ether, is the unit of Ethereum's currency) is used as a baseline for on-chain cost in this work, which always consumes 21,000 gas.} on the Ethereum platform. It is obviously impractical.
\item {\em Overhead of off-chain execution.} Other technologies \cite{truebit,arbitrum,yoda,ace,plasma} orchestrate off-chain execution and contract arbitration, and thus significantly reduce the on-chain cost.
However, to guarantee an alignment and the consistency between the off-chain execution and the contract arbitration, they all rely on restricted custom virtual machines (VMs) for the off-chain computations. For example, floating-point arithmetic and the multi-threaded execution, which commonly exist in DNN computations, are prohibited in these VMs. This may result in orders of magnitude slowdown compared to the native execution. It also significantly increases the burden for the verifiers and the pinpoint protocol.
\squishend

\para{Problem statement.} The problem we consider in this paper is: \textit{how to enable smart contract verification of DNN computations with low cost, both on-chain and off-chain}. Regarding the scale of the computations, we consider the popular DNN models, such as ResNet50 \cite{resnet}, VGG16 \cite{vgg} and MobileNet \cite{mobilenet}.\footnote{We focus on evaluation of DNN models (i.e. inference), and leave DNN training as future work.}
Once the smart contract obtains this scale of DNN competence, one can imagine many types of real ``smart'' applications, such as:
\squishlist
\item {\em Intelligent Automated Market Maker (AMM).} AMM such as Uniswap \cite{uniswap} has revolutionized the economic ecosystem of Ethereum. DNN can improve the liquidity and profit of the market maker \cite{amm1,amm2}, which is also envisioned in the DeFi community \cite{iamm}.
\item {\em Decentralized AI marketplace.} Decentralized exchanging of digital goods has been proposed recently for stronger fairness, compared to those centralized solutions \cite{fairswap,optiswap,zkcp}. There is a clear need to exchange AI models in the marketplace. Verification of DNN computations is essential to enable this scenario \cite{algorithmia}.
\item {\em Blockchain-based Uber (BUber).} BUber has been extensively discussed to mitigate the fairness issues due to the opacity of the centralized intermediary \cite{lazooz,book}. Scheduling algorithm with a DNN capability can further enhance BUber's efficiency \cite{uber}.
\item {\em Decentralized paper ballot counting.} Handwritten-signature recognition is used for paper ballot counting in national elections \cite{ballot}. Allowing a decentralized smart contract to run the process may help increase the transparency and make the public more confident about the results \cite{vote,vote2}.
\squishend

\para{Agatha system.}
In response to this community vision, we develop a system named Agatha, which demonstrates the first practical contract verification for DNN computation on the public Ethereum. Agatha follows the existing off-chain execution approaches with contract arbitration \cite{yoda,plasma,truebit,ace}, so the on-chain overhead is greatly reduced. The main difference between Agatha and the previous approaches is about the off-chain computation, where Agatha enables the {\em native execution of DNN computation}, whereas others only support restricted VM-based execution. Therefore, Agatha significantly reduces the overhead of the off-chain execution, making the complex computation practical.

\para{Technologies.} Enabling the off-chain native execution of DNN computation for smart contract is our main contribution. It faces two significant challenges:
\squishlist
\item{\em Misalignment of native execution and contract arbitration.} Although previous work, such as TrueBit \cite{truebit}, claims to achieve native off-chain execution for general-purpose computation, DNN computation is much different from their showcased applications. DNN computation is implicitly expressed as a computation graph of operations with rich toolchain support \cite{tensorflow,onnx,onnxruntime,pytorch}. In contrast, existing contract arbitration requires the computation to be expressed by serialized VM instructions. Therefore, considering hardware features used by conventional DNN computations, such as multi-threading, SIMD or even GPU instructions, it is an enormously complicated (and unnecessary) detour to transcode between DNN computations and serialized VM instructions and ensure their consistency.
\item{\em Cross-platform inconsistency.} Native execution of DNN computation, especially on different hardware platforms, may lead to different execution results (floating-point vectors). This can be ascribed to various factors, such as imprecise approximations, different accumulation orders of floating-point numbers, etc. Different execution results or temporary variables during the execution can discredit the fraud proof of the quorum since every honest node can be a ``fraud'' due to execution on different platforms.

\squishend

To solve the challenges, the key technologies of Agatha system are twofold:
\squishlist
\item {\em Graph-based Pinpoint Protocol (GPP)}, which is a protocol that can generate the ``fraud proof'' interactively and efficiently. Instead of representing DNN computation as executing VM instructions, GPP represents DNN computation as evaluating a graph of operations (e.g., Conv, Gemm) and further a graph of basic operations (e.g., fadd, fmul), where the evaluation results of basic operations can be efficiently arbitrated by our smart contract equipped with the floating-point arithmetic emulation. Moreover, the graph-based representation is highly compatible with conventional toolchains of DNN computation. By guaranteeing that evaluating the graph is consistent with the native execution, Agatha's pinpoint protocol is purely based on the graph evaluation, which circumvents the misalignment challenge. In addition, we introduce the {\em Two-phase Pinpoint} design, which optimizes the protocol to equip an orchestrated execution of both coarse-grained evaluator (native execution in the granularity of operation) and fine-grained evaluator (simulation in the granularity of basic operation).
\item {\em Cross-evaluator consistent execution (XCE)},
which guarantees the three consistencies between: (1) native execution and operation evaluation,  (2) operation evaluation and the evaluation of basic operation, and (3) basic-operation evaluation and the smart contract arbitration. To achieve every consistency, we have two steps in general: the first is to make the execution either compliant with the IEEE-754 standard of floating-point arithmetic or in integers; the second is to ensure a fixed order of floating-point sum/product. Following the workflow, we conduct a comprehensive investigation and solid tests to ensure consistency, which spans over hardware heterogeneity, compiler options, arithmetic libraries and so on. We formally prove that with XCE, the cross-platform consistency and the correctness of Agatha system are achieved.

\squishend

In summary, GPP establishes the infrastructure that bridges the gap between the native DNN computation and the contract arbitration. And XCE forms the correctness foundation of GPP by ensuring the cross-evaluator consistency, which naturally leads to the cross-platform consistency.

\para{Results.}
We showcase Agatha system for DNN inference using MobileNet \cite{mobilenet}, ResNet50 \cite{resnet} and VGG16 \cite{vgg}, which are models widely used in the industry.
The evaluations confirm the correctness of XCE, using unit tests and end-to-end tests, which include $\sim$7500 floating-point corner cases.
We also measure the performance and the gas consumption of Agatha. Our off-chain native execution is observed to have an insignificant latency overhead (3.0\% on average) compared to the original DNN computation. By contrast, previous VM-based off-chain execution has an average slowdown of 620$\times$, which is 602$\times$ greater than ours. Regarding the on-chain cost: When there is no dispute, i.e., the normal case, the cost of Agatha is equivalent to $\sim$3 ETH-transfer transactions; When the two disputing parties fight all the way to the arbitration, i.e., the worst case, the cost is equivalent to $\sim$86 ETH-transfer transactions. These performance and cost numbers demonstrate Agatha's practicality for DNN computation, given both on-chain and off-chain overhead.

%% file: pre.tex
\section{Smart Contract for Complex Computations}\label{sec:pre}

In this section, we give the background about smart contract for complex computation, which is the approach Agatha follows.

As mentioned in the introduction, our research is inspired by off-chain execution approaches with contract arbitration, such as Truebit \cite{truebit}, Arbitrum \cite{arbitrum}, YODA \cite{yoda} and ACE \cite{ace}.
Figure \ref{fig:ABV} shows the essence of the approach. The base is the smart contract, which runs on a peer-to-peer network of a huge number of nodes.
Because the on-chain execution needs to always maintain a worldwide consensus, it cannot afford to run expensive computations. The basic idea of previous approaches is to have a small number of off-chain nodes (e.g., 10 independent parties or a dynamic quorum). The off-chain nodes can be elected either by network reputation or registration via deposit on the smart contract.

Figure \ref{fig:ABV} shows the steps in these approaches. For the initialization step \circled{0}, the requestor sends the computation task $\phi$ and input data D to the off-chain nodes while optionally making corresponding commitment\footnote{The requestor implicitly exists as the transaction senders while the commitment is omitted in the aggregated execution of transactions \cite{optimism}.}. In step \circled{1}, the first node that finishes computation, which is called the submitter in this paper, claims ``$\mathcal{R}=\phi(\mathcal{D})$''. This means that $\mathcal{R}$ is the result of computation $\phi$ on input $\mathcal{D}$. $\phi$ is too expensive for the smart contract to re-execute, so it is only submitted to the verifiers, which are the rest of the off-chain nodes. In step \circled{2}, every verifier independently validates the claim. Suppose the lower-left verifier declares that the claim is wrong, it starts a \textit{pinpoint protocol}, involving the verifier itself, the submitter and the smart contract. This is shown in step \circled{3}. The verifier tries to disprove the claim by pinpointing one concrete erroneous step in it, with the contract arbitrating in step \circled{4}. Obviously, arbitrating about a single step is easy for the smart contract. In the end, if the submitter survives all challenges from verifiers for a pre-defined period $T^v$, ``$\mathcal{R}=\phi(\mathcal{D})$'' is accepted by the smart contract.

\begin{figure}[th]
  \centering
       \vspace{-10pt}
  \includegraphics[width=0.40\textwidth]{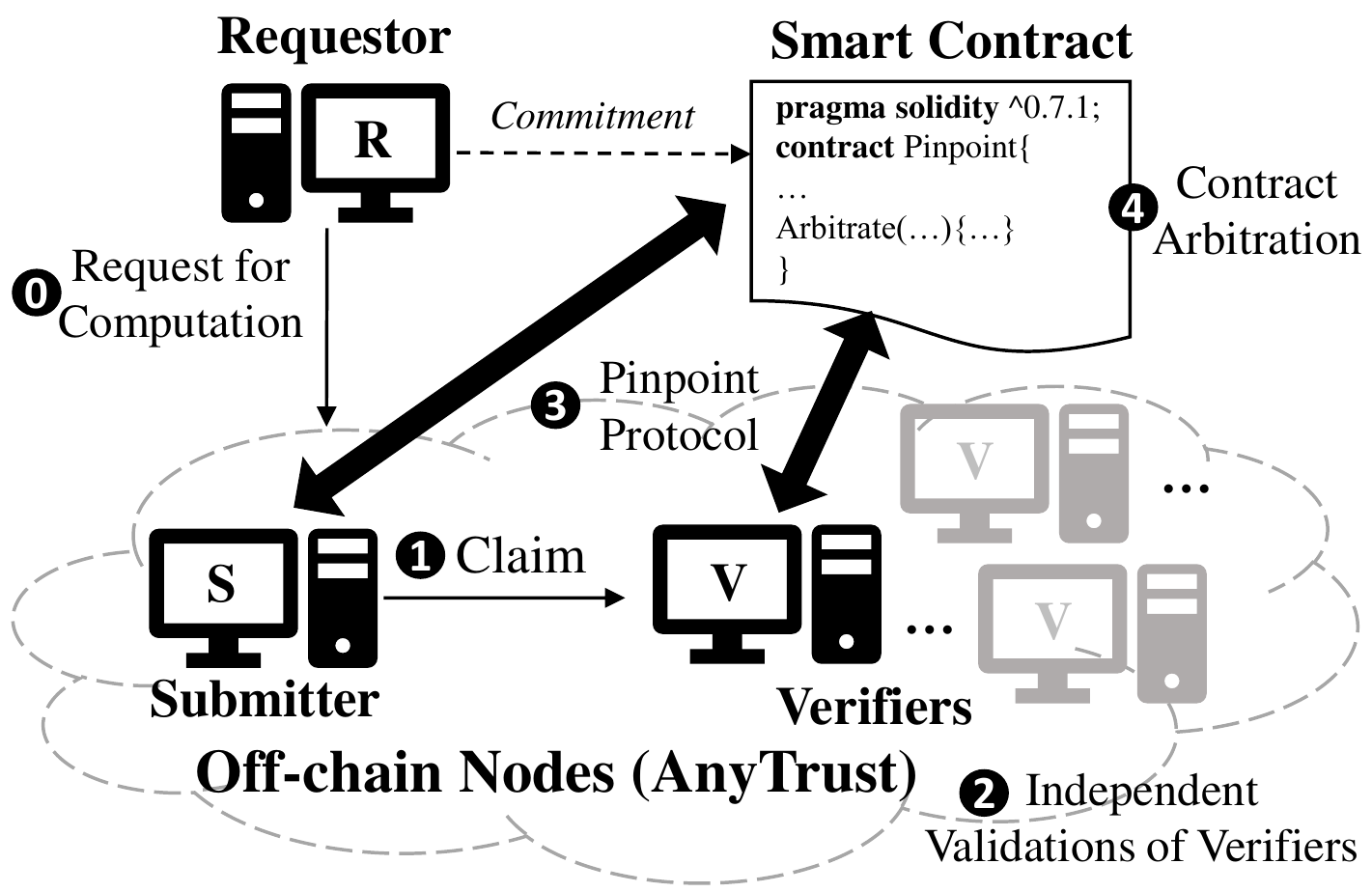}
  \vspace{-5pt}
  \caption{The essence of the previous approaches} \label{fig:ABV}
  \vspace{-5pt}
\end{figure}

There are various assumptions about the verifiers to guarantee security and liveness, such as AnyTrust \cite{arbitrum}, Financially Rational Players \cite{truebit} and the hybrid Byzantine assumption \cite{ace,yoda}. For simplicity, we use the AnyTrust assumption from Arbitrum \cite{arbitrum}.

\para{The AnyTrust assumption.} AnyTrust, rather than ``majority trust'', assumes that, for every claim, there is at least one honest node. Namely, either the submitter is honest, or at least one verifier is honest and will challenge within the pre-defined period. Suppose there are $m$ verifiers, $m-1$ of whom collude to stay silent about a submitter's wrong claim. Under AnyTrust, the only honest verifier will still succeed in disproving the wrong claim in front of the smart contract, and the wrong claim is thus rejected (formal definitions in Section \ref{sec:security}). Like the previous work \cite{arbitrum, truebit}, we also assume data availability  (i.e., $\mathcal{D}$ and $\phi$ should be accessible to all verifiers) and anti-censorship (i.e., every verifier can always interact with the contract). They are solved by orthogonal countermeasures \cite{ipfs, eclipse}.

\para{The pinpoint protocol.} The pinpoint protocol\footnote{Equivalent terminologies include {\em Bisection protocol} \cite{arbitrum} and {\em Verification Game} \cite{focs,truebit}.} also needs more explanations.  %
The goal of the protocol is shown in Figure \ref{fig:pinpoint}. The computation $\phi$ consists of a sequence of VM instructions (denoted as VMI$_1$, VMI$_2$, $\dots$, VMI$_n$) and the initial VM state is S$_0$. However, the verifier and the submitter get different states S$_n \neq$ S$'_n$. The goal of the protocol is to pinpoint VMI$_k$, such that S$_{k-1} =$  S$'_{k-1}$ but S$_k \neq$  S$'_k$, and to send (VMI$_k$, S$_{k-1}$, S$_k$) to contract for arbitration. The submitter and the verifier are forced to get the same $k$ with a challenge-response mechanism with a timeout penalty \cite{arbitrum,fairswap,truebit}.
The pinpoint protocol is very efficient: for time complexity, the verifier and the submitter only need $O(\log n)$ rounds of challenge-response for both parties to commit to the number $k$, and the contract only needs a constant time to arbitrate the disagreement about VMI$_k$; regarding the space complexity, the state is structured as a Merkle tree~(MT) \cite{mt}, corresponding to a logarithmic-size message to the contract.

\begin{figure}[t]
  \centering
  \includegraphics[width=0.42\textwidth]{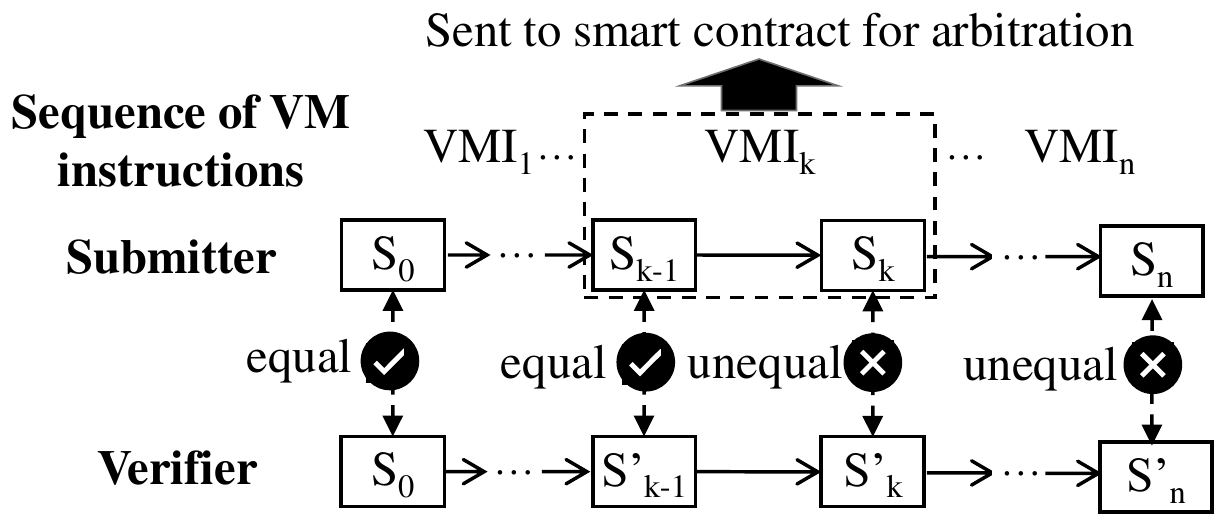}
  \vspace{-10 pt}
  \caption{Pinpointing a disputed step in VM execution} \label{fig:pinpoint}
  \vspace{-15pt}
\end{figure}

It is worth emphasizing that two conditions are needed for the pinpoint protocol to work: (1) The pinpoint mechanism must be able to locate the disputed computation step, and the contract must be able to arbitrate it efficiently. (2) For every step of the computation, there is only one correct output for a given input. Otherwise, the contract cannot verify its correctness. To achieve the two conditions, the computations handled by previous technologies are written as sequential VM
instructions, which fits the instruction-register-memory paradigm. In addition, features such as multi-threading and floating-point are eliminated to guarantee the cross-platform consistency. For example, Arbitrum VM is based on EVM \cite{evm} while TrueBit leverages a restricted WebAssembly \cite{wasm}.

\para{The overhead of off-chain computation.} Besides the overhead of contract execution (i.e., on-chain part), the overhead of off-chain execution is also critical since the off-chain overhead determines the workload and throughput of off-chain nodes. For example, submitter or verifiers in TrueBit \cite{truebit} earn bounties proportional to the length of VM instructions, which is the cost to post the computation task for the requestor; Arbitrum VM leverages the Intel SHA extension \cite{sha} to achieve a high off-chain throughput. Furthermore, in DeFi applications, the off-chain latency directly determines the efficiency of the market because of the off-chain racing counterparties (e.g., liquidity providers \cite{balance} and traders \cite{trader}).

%% file: overview.tex
\section{Agatha Overview} \label{sec:overview}

The Agatha system consists of a smart contract deployed on Ethereum and several independent {\em Agatha clients}. For a DNN computation task, the Agatha system focuses on the verification process after one Agatha client submits a result of the task (specifically, the hash value of the result). If any client finds the result incorrect, and then wins the challenge-response dispute following the Agatha pinpoint protocol, the contract will reject the submitter's result. Otherwise, the submitter's result is accepted.

Similar to the previous technologies, Agatha has a workflow to reduce the on-chain overhead. However, the key difference is that Agatha enables the native DNN computation which reduces the off-chain overhead significantly. The design of the Agatha system is shown in Figure~\ref{fig:overview} and explained in this section.

\begin{figure}[h]
  \centering
  \includegraphics[width=0.46\textwidth]{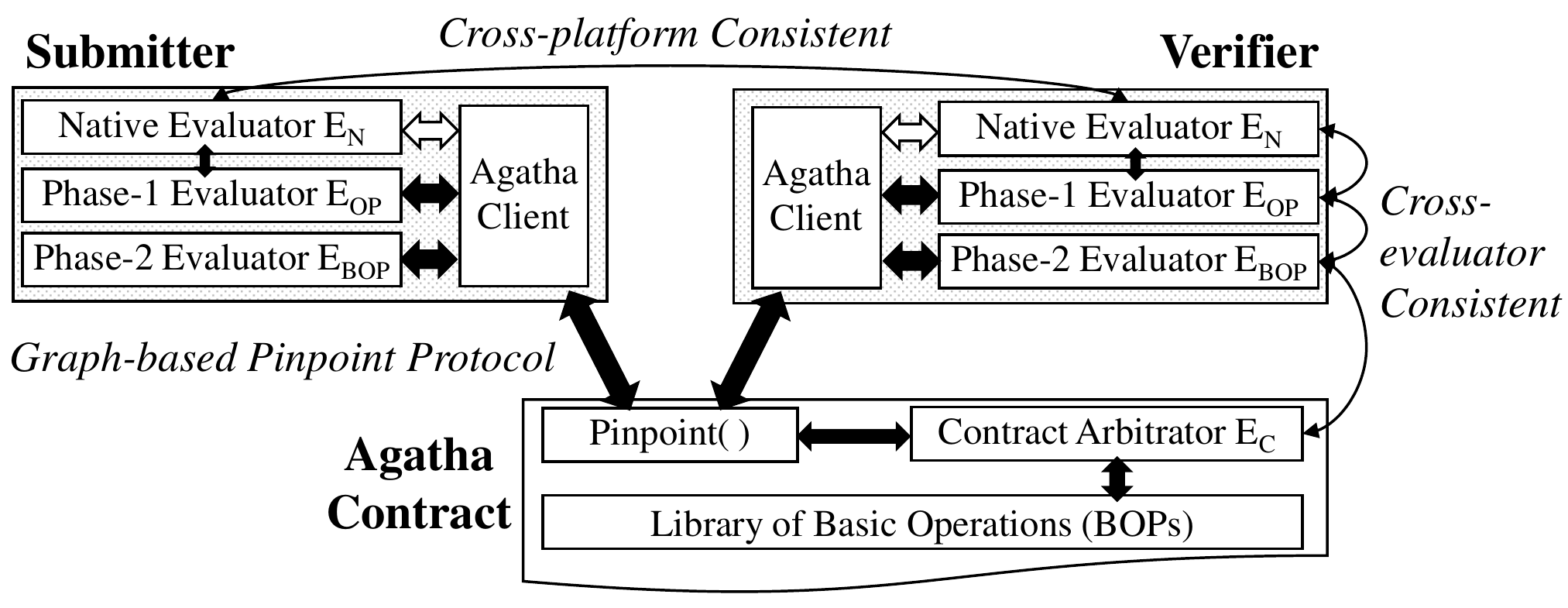}
  \caption{The overview of Agatha system} \label{fig:overview}
\end{figure}

\subsection{Native execution}

\begin{table}[b]
\centering
\footnotesize
\setlength{\tabcolsep}{10pt}
\renewcommand{\arraystretch}{1.05}
\begin{tabular}{@{}lccc@{}}
\toprule
\textbf{Steps} & \textbf{Previous} & \textbf{Hypothetical} & \textbf{Agatha} \\ \midrule
Submitter's execution & VM & Native & Native \\
Verifiers' validation & VM & Native & Native \\
Pinpoint protocol & VM & VM & Semi-native \\
Contract arbitration & VMI & VMI &  BOP \\ \bottomrule
\end{tabular}
\caption{\protect\centering Comparison of the previous work, a hypothetical method and Agatha}
\label{tab:overview}
\end{table}

The core innovation of Agatha is enabling native execution in scalable contracts for DNN computation  (see Table \ref{tab:overview}). There is a native evaluator $E_N$ in each client. In Agatha, both submitter and verifier execute DNN computation with multi-threaded, highly optimized and hardware-accelerated $E_N$, instead of a restricted single-threaded VM. Using $E_N$ significantly reduces both the submitter's execution latency and the verifier's validation latency. For example, a multi-threaded $E_N$ only takes 0.076 seconds to run a VGG16 inference (see Section~\ref{sec:perf}); an ideal VM (simulated by a single-threaded, highly optimized native evaluator) takes 0.439 seconds, 5.78 times more than the $E_N$;  restricted VMs such as Arbitrum VM are typically 800 times slower than the ideal VM. If $E_N$ is further accelerated by GPUs, its performance improvement over restricted VMs can be three orders of magnitude.

Although some previous VM-based schemes have imagined combining native execution, VM-based pinpoint protocol and arbitration~\cite{truebit}, they provide no detail. We consider this highly difficult, because of the paradigm misalignment%
. To enable pinpoint protocol and contract arbitration, the VM paradigm takes a sequence of VM instructions and an initial state as inputs, and outputs a final state. However, a multi-threaded $E_N$ takes a function and some data as inputs, and outputs the result. If we want to bridge them, there are two potential routes. First, if we define $\phi$ as VM instructions, we need to transcode the serialized VM instructions to a native multi-threaded program and ensure their consistency. Second, if we define $\phi$ as a function, the verifier and the submitter must have consensus on how $\phi$ is represented by a specific sequence of VM instructions. However, both routes are hugely unnecessary detours. As we will explain next, the graph-based paradigm is a far more direct representation of DNN computation. %

\subsection{Unified graph-based paradigm}

Agatha forgoes the VM-based paradigm, which is a wrong abstraction for DNN computation. A DNN, as a function, says nothing about instructions, registers, memory, etc. Alternatively, following the mainstream paradigms~\cite{onnx, tensorflow, pytorch}, we express a DNN computation as a computation graph consisting of \textit{operations}. A computation graph is a directed acyclic graph with operations as its nodes and tensors as its edges. For example, Figure \ref{fig:onnx} shows the computation of ResNet50~\cite{resnet}, and the internal computations of \textsf{Conv}, \textsf{Relu} and \textsf{Gemm} are shown in the figure. An operation is defined mathematically using \textit{basic operations}, such as summation, max and multiplication. In the rest of the paper, we use the terminology ``basic operation'', or {\em BOP}, to refer to an individual computation of numbers. The terminology ``operation'' or {\em OP}, without a preceding ``basic'', means the computation of tensors.

\begin{figure}[thb]
  \centering
  \vspace{-10pt}
  \includegraphics[width=0.43\textwidth]{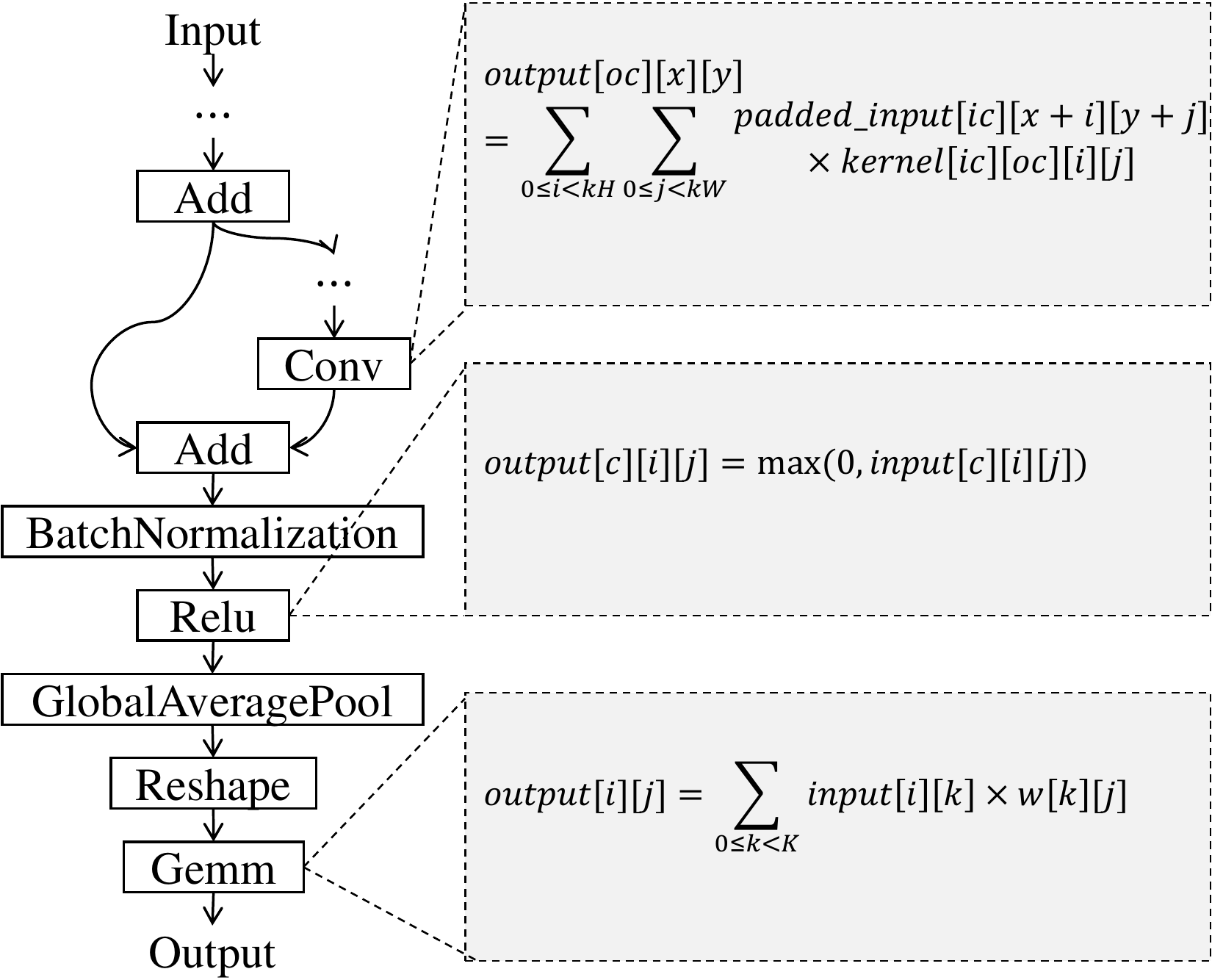}
  \vspace{-5pt}
  \caption{The computation graph of ResNet50} \label{fig:onnx}
  \vspace{-10pt}
\end{figure}

Using the graph representation provides a unified paradigm from native execution to contract arbitration. It also has two main advantages. First, we can leverage existing high-performance graph evaluators, such as ONNX Runtime~\cite{onnxruntime}, TensorFlow~\cite{tensorflow} and PyTorch~\cite{pytorch}, as our $E_N$. Second, as we will show next, we can design a simple and efficient pinpoint protocol for DNN computation.

Each Agatha client leverages a phase-1 evaluator $E_{OP}$ and a phase-2 evaluator $E_{BOP}$, used for our two-phase pinpoint protocol. In phase 1, the DNN computation in dispute invokes $E_{OP}$ to compute the results of the operations, and then the Agatha contract can locate an operation in dispute. In phase 2, the operation in dispute further invokes $E_{BOP}$ to generate the results of the basic operations, and thus the contract can locate a basic operation in dispute. After pinpointing, the contract performs the arbitration by calculating the basic operation on Ethereum with its on-chain evaluator $E_C$.

The simplicity comes from the graph representation. For this protocol, the arbitration contract only needs to understand integer and floating-point arithmetic, but not other complicated machinery; the client only needs to understand computational graphs, but not sophisticated VM instructions and states.

The efficiency benefits from both native execution and the nature of DNN computation. $E_{OP}$ can invoke the high-performance $E_N$, and $E_{BOP}$ can make arbitration simple. Imagine two alternative pinpoint protocols. The first one, without phase 1, uses $E_{BOP}$ to generate internal intermediate results of the whole graph, and lets the contract locate a basic operation in dispute. This is too inefficient for a DNN computation, which usually has billions of basic operations. The second one, without phase 2, demands the contract to directly locate an operation in dispute, and then arbitrate the operation by re-executing it on-chain. This is not practical for DNN computations, whose operations are too complex to implement and execute in smart contracts.

\subsection{Consistency}
\label{sec:consistency}
Forgoing the VM-based paradigm brings a new issue about verification -- inconsistency. Using the graph-based paradigm, we notice that operations in DNN computation are only defined as mathematical formulas, but the details of the implementations still have flexibility, such as the precision and the order of internal basic operations. Consequently, with the same input, the output of DNN computation, usually a floating-point vector representing confidence values, is not strictly consistent. For example, we test one \textsf{Gemm} operation with the same inputs on different software/hardware platforms. The results are shown in Table \ref{tab:inconsistent}. The tests in Table \ref{tab:inconsistent}(a) run the \textsf{Gemm} operation with different $E_N$ on an i7-8700 processor with 12 logical cores. ``Serialization'' refers to the straightforward for-loop implementation of matrix multiplication, similar to our $E_{BOP}$. Others are popular software for DNN computation. We can see that their outputs are different from each other (i.e., with different hash values). More interestingly, we test PyTorch and find that even the same software on different hardware yields different results, as shown in Table \ref{tab:inconsistent}(b).

\begin{table}[thb]
    \centering

  \begin{subtable}[b]{0.21\textwidth}
    \caption{Different software} \label{tab:runtime}
    \footnotesize
    \setlength{\tabcolsep}{4pt}
 \renewcommand{\arraystretch}{1.1}
    \begin{tabular}{lc}
\hline
Software     & Result Hash \\ \hline
NumPy       & \texttt{068754...}     \\
ONNX Runtime & \texttt{c2baad...}    \\
PyTorch     & \texttt{630fa0...}    \\
TensorFlow  & \texttt{50bd3a...}    \\
Serialization   & \texttt{c415be...}    \\
\hline
\end{tabular}
  \end{subtable}
  \hspace{12pt}
  \begin{subtable}[b]{0.2\textwidth}
    \caption{Different hardware}     \label{tab:hardware}
    \footnotesize
    \setlength{\tabcolsep}{3pt}
 \renewcommand{\arraystretch}{1.1}
    \begin{tabular}{lc}
\hline
Hardware              & Result Hash \\ \hline
4 vCPU (i7-8700)  & \texttt{fd8858...}    \\
8 vCPU (i7-8700)  & \texttt{c53a15...}    \\
12 vCPU (i7-8700) & \texttt{630fa0...}    \\
16 vCPU (Skylake)     & \texttt{068754...}     \\
16 vCPU (Broadwell)   & \texttt{f71df1...}    \\

\hline
\end{tabular}
  \end{subtable}
  \caption{\protect\centering The inconsistency among the results of the same \textsf{Gemm} operation on different platforms}
  \label{tab:inconsistent}
\end{table}

With this issue in mind, we can revisit the pinpointing and arbitration mechanism and understand two major challenges about verification -- cross-platform consistency and cross-evaluator consistency. Without cross-platform consistency, the verifiers are unable to tell if a submitted claim is correct by re-executing it. On smart contracts, the equality of two large data objects is checked by comparing their secure-hash values or Merkle tree roots. Equality means an exact match -- even one-bit difference will fail the equality test. Even if someone would invent an equality test that could tolerate a degree of imprecision, it would not satisfy the requirement, as we observe that the imprecision of some intermediate operations cannot be bounded. For example, the reciprocal of a small value is very sensitive to its imprecision. Therefore, we must tackle the challenge directly, so that every $E_N$ can make results consistent down to every bit.

Cross-evaluator consistency guarantees the correctness of honest clients in our pinpoint protocol and naturally leads to the cross-platform consistency (details in Section \ref{sec:security}). Specifically, for a computational graph, $E_N$ and $E_{OP}$ should output a consistent result; for an operation, $E_{OP}$ and $E_{BOP}$ should output a consistent result; for a basic operation, $E_{BOP}$ and $E_C$ should output a consistent result. Like cross-platform consistency, our goal is to identify and eliminate all sources of inconsistency among the four evaluators.

%% file: pinpoint.tex
\section{Graph-Based Pinpoint Protocol}\label{sec:pinpoint}

As mentioned earlier, pinpoint protocol is the bridge between native evaluator $E_N$ and contract arbitrator $E_C$. To align $E_N$ and $E_C$, Agatha does not design the pinpoint protocol based on VM, which is not aligned with the intrinsic nature of DNN computation. Instead, we design a graph-based pinpoint protocol (GPP). GPP includes the following three components and defines how to use them for the pinpoint protocol: (1) Phase-1 evaluator $E_{OP}$ which aligns with $E_N$ and executes in the granularity of operation; (2) Phase-2 evaluator $E_{BOP}$ which aligns with both $E_C$ and $E_{OP}$, and executes in the granularity of basic operation (BOP), along with the setup tools for $E_{BOP}$; (3) an enhanced evaluator $E_{C}$ which can arbitrate the floating-point arithmetic in BOPs.

Next, we begin with explaining the components for $E_{BOP}$, because $E_{BOP}$ plays a central role in GPP. It not only determines the BOPs that $E_C$ needs to support, but also is the basis of $E_{OP}$.

\subsection{Enabling graph-based pinpoint with $E_{BOP}$} \label{sec:bop}

\para{Circuit generation.} Representing a DNN computation as a computation graph can be considered in two levels: the operation (i.e., OP)  layer and the basic operation (i.e., BOP) layer, as discussed earlier. The first-level computation graph is already generated explicitly \cite{onnx}, similar to the one in Figure \ref{fig:onnx}. However, we need the ability to generate the second-level computation graph (i.e., circuit). Table \ref{tab:bop} shows the BOPs we support. They are in four categories: floating-point arithmetic, integer arithmetic, assignment and typecasting. We have confirmed that every operation in our runtime can be expressed as a circuit with these BOPs.

\begin{table}[thb]
\centering
\footnotesize
\setlength{\tabcolsep}{7pt}
\renewcommand{\arraystretch}{1.2}
\begin{tabular}{ll}
\hline
{\bf Categories}    & {\bf Basic Operations }                                                                                                \\ \hline
\begin{tabular}[c]{@{}l@{}} Floating-point \\Arithmetic \end{tabular}& \begin{tabular}[c]{@{}l@{}}f32\_add, f32\_sub, f32\_mul, f32\_div, f32\_min, \\ f32\_max, f32\_sqrt, f32\_round, f32\_floor\end{tabular} \\ \hline
Integer Arithmetic        & i32\_add, i32\_sub, i32\_mul                                                                                                 \\ \hline
Assignment          & =   \hspace{20pt} ( the operand type can be f32, i32 or u8)    \\ \hline
Type Casting        & \begin{tabular}[c]{@{}l@{}}f32\_to\_u8,  u8\_to\_f32, u8\_to\_i32, \\ i32\_to\_f32, f32\_to\_i32\end{tabular}          \\ \hline
\end{tabular}
\caption{Eighteen basic operations (BOPs) in this work}
\vspace{-20pt}
\label{tab:bop}
\end{table}

We implement a rule-based circuit generator for conventional OPs in DNN computation. The generator handles the operations with different shapes and other parameters. We have confirmed that the generated circuits and the operations are {\em logically} identical, but Section \ref{sec:determine} will discuss the inconsistency situations when {\em concretely} evaluating them.
Figure \ref{fig:circuit} shows a tiny example of the generated circuit for an integer matrix multiplication (i.e., \textsf{MatMulInteger}). Figure \ref{fig:circuit}(a) shows the multiplication, and Figure \ref{fig:circuit}(b) shows the circuit. The circuit is a directed acyclic graph (DAG). The vertices of the graph are BOPs (e.g., i32\_add and i32\_mul, shown as ``+'' and ``$\times$''). Two special vertices, shown as
``in'' and ``out'', are the source and the sink of the graph. The directional wires represent the variables of the computations, including the input variables (i.e., in$_1\sim$ in$_6$), intermediate variables (i.e., v$_1\sim$ v$_4$), and the output variables (i.e., out$_1$, out$_2$). %
We also use v$_5$ and v$_6$ as the aliases of out$_1$, out$_2$, respectively.

Since DNN computation is already expressed as a graph of operations, with the circuit generator, we can represent the whole DNN computation as a giant circuit, which involves billions of BOPs.

\begin{figure}[thb]
  \centering
    \vspace{-5pt}
  \includegraphics[width=0.46\textwidth]{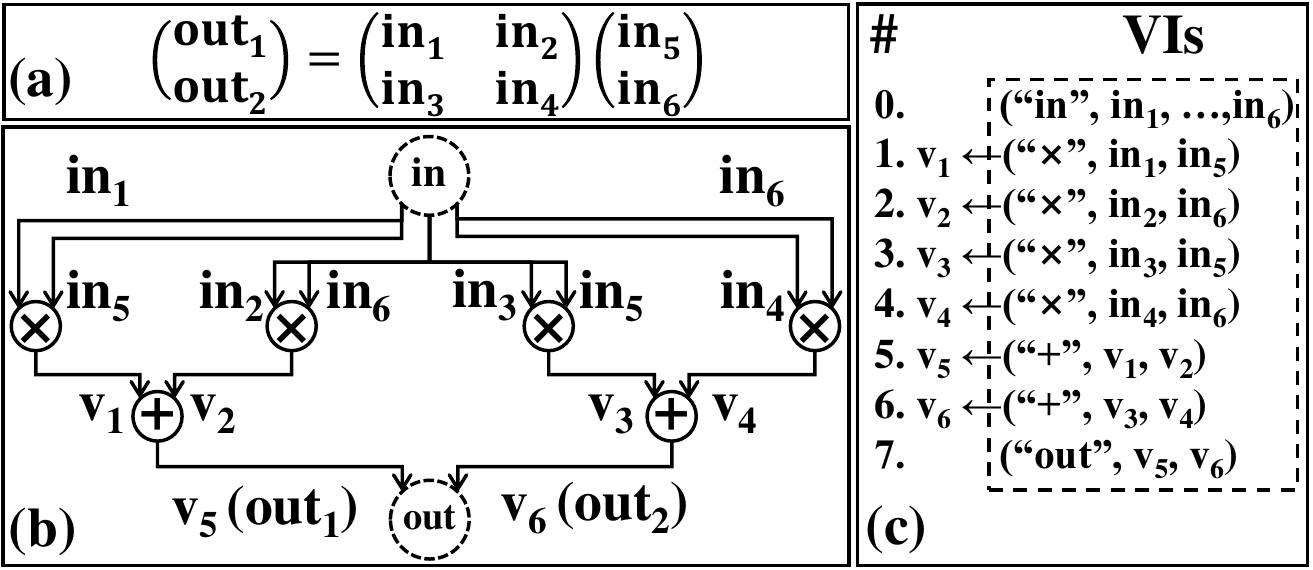}
   \vspace{-10pt}
  \caption{\protect\centering A circuit of matrix multiplication: (a) the formula, (b) the circuit, (c) serialized Vertices with Inputs (VIs).} \label{fig:circuit}
  \vspace{-5pt}
\end{figure}

\para{Circuit serialization.} Because the computation graph $\textrm{G}=(\textrm{V},\textrm{E})$ is a DAG, a topological sorting (TS) \cite{ts} can be used to serialize it, which has a complexity of $O(|\textrm{V}|+|\textrm{E}|)$. Figure \ref{fig:circuit}(c) shows the serialization result of the graph in Figure \ref{fig:circuit}(b). Every item in the sequence, except the first and the last, is denoted by its BOP and the variables on its incoming edges. We call each of these items a {\em ``vertex with inputs''}, or {\em VI}.

A VI has two forms: the symbolic form and the concrete form. In the symbolic form, referred to as VI$^S$, every variable is assigned a unique name by the circuit generator, such as in$_3$ and v$_4$. For example,  (``$\times$'', in$_1$, in$_5$) can be a VI$^S$. In the concrete form, referred to as VI$^C$, every variable is replaced by its actual value. For example,  (``$\times$'', 1, 5) can be a VI$^C$.

Therefore, the pinpoint problem can be defined using VI$^S$ and VI$^C$ as follows. {\em Prerequisite}: the submitter and the verifiers agree on the same VI$^S$-sequence and all input values. {\em Problem}: the submitter and a verifier get different VI$^C$-sequences, so the smart contract needs to determine who is incorrect (note that it does not imply that the other is correct).

\para{Circuit evaluator $E_{BOP}$.} $E_{BOP}$ produces the VI$^C$-sequence based on the VI$^S$-sequence and the concrete values of the inputs. It sequentially outputs every VI$^C$ using the corresponding VI$^S$. Let's assume (in$_1$, ..., in$_6$) = (1, ..., 6) in Figure \ref{fig:circuit}(a). First, VI$^C_0$ is initiated as (``in'', 1, ..., 6). Then, the evaluations for VI$^C_1$, ..., VI$^C_6$ are performed, based on VI$^S_1$, ..., VI$^S_6$. The results of these VI$^C$s are stored. For example, VI$^C_5$ refers to v$_1$ and v$_2$, so the evaluator loads the output of VI$^C_1$ (i.e., v$_1$ = 5) and VI$^C_2$ (i.e., v$_2$ = 12), respectively. The process continues until the last VI$^C$, i.e., VI$^C_7$ = (``out'', 17, 39) is produced.

\para{Commitment scheme.} With the setup work of circuit generation, serialization and the evaluator, when a verifier disagrees about the computation result of the submitter, the pinpoint protocol begins. In this subsection, we assume the whole DNN computation is expressed as a circuit. Then, our pinpoint protocol can be based on a commitment scheme \cite{commitment} leveraging Merkle trees (MTs) and a smart contract. The protocol allows the submitter and the verifier to locate their leftmost divergent VI$^C$. The complexity is $O(\log n)$, where $n$ is the length of the VI sequence. The commitment scheme guarantees the non-repudiation of the disputing parties during the process. The protocol includes one preparation step and three interaction steps, as intuitively shown in Figure \ref{fig:merkle} (formal settings in Section \ref{sec:security}) :
\squishlist
\item{\em Prepare.} Before the pinpoint, with the VI$^S$-sequence as leaves, both the submitter and the verifier generate and reach a consensus on MT$^{S}$. They also agree on the input values (i.e., VI$_0^C$). For example, the consensus can be denoted as a unanimous signature to the contract on the root of MT$^{S}$ and VI$_0^C$. Then, the submitter submits the output of the circuit (e.g., VI$_7^C$) to the contract.

\item {\em Locate the leftmost divergent VI$^C$.} The verifier challenges the submitter by invoking the contract API if the verifier disagrees with the output. The submitter is first required to generate MT$^{C}$ from the VI$^C$-sequence. Then, according to the on-chain challenges of the verifier, the submitter responds to the contract with a complete Merkle path traversing from the treetop to the leaf. The verifier selects the path based on the local MT$^{C}$ in order to find the leftmost VI$^C$ divergent from the submitter's. During the process, the contract keeps validating if the Merkle path of the submitter is correct and follows verifier's challenges. Thus, the submitter cannot cheat to regenerate an MT$^{C}$ while the verifier cannot deny the challenge requests sent to the contract.\footnote{If the submitter aborts the process, the verifier invokes the timeout function of the contract to mark the submitter as ``incorrect''.} Suppose that the VI$_5^C$ is located, which means that VI$^C$s before the VI$_5^C$ are consistent.\footnote{If the submitter and the verifier locate the leftmost divergent VI$^C$ as VI$_0^C$, our protocol still works since VI$_0^C$ has a unanimous commitment.} By referring to the definition of VI$_5^C$, the submitter and the verifier dispute either about the concrete input (i.e., the value of v$_1$ or v$_2$) or about the BOP (i.e., ``+'').

\item {\em Upload the MT$^{S}$ path.} If the disagreement is about the BOP, the submitter is required to upload a Merkle path of VI$_5^S$ to the contract. The contract verifies the path and checks if the BOP in VI$_5^S$ matches the submitter's. If both pass, the verifier is marked ``incorrect'' by the contract.

\item{\em Upload the MT$^{C}$ path.} If the disagreement is about the input (e.g., v$_2$), the submitter is required to upload: (1) an MT$^{S}$ path of VI$_5^S$ to prove that the output of VI$_2^C$ provides the disputed input; (2) an MT$^C$ path of VI$_2^C$ to prove that submitter's value for v$_2$ is correct. If the contract validates the submitter's uploaded paths and that the output of VI$_2^C$ is submitter's v$_2$, the verifier is marked as ``incorrect'' by the contract. Otherwise, the submitter is marked ``incorrect''.

\squishend

\begin{figure}[thb]
  \centering
  \includegraphics[width=0.43\textwidth]{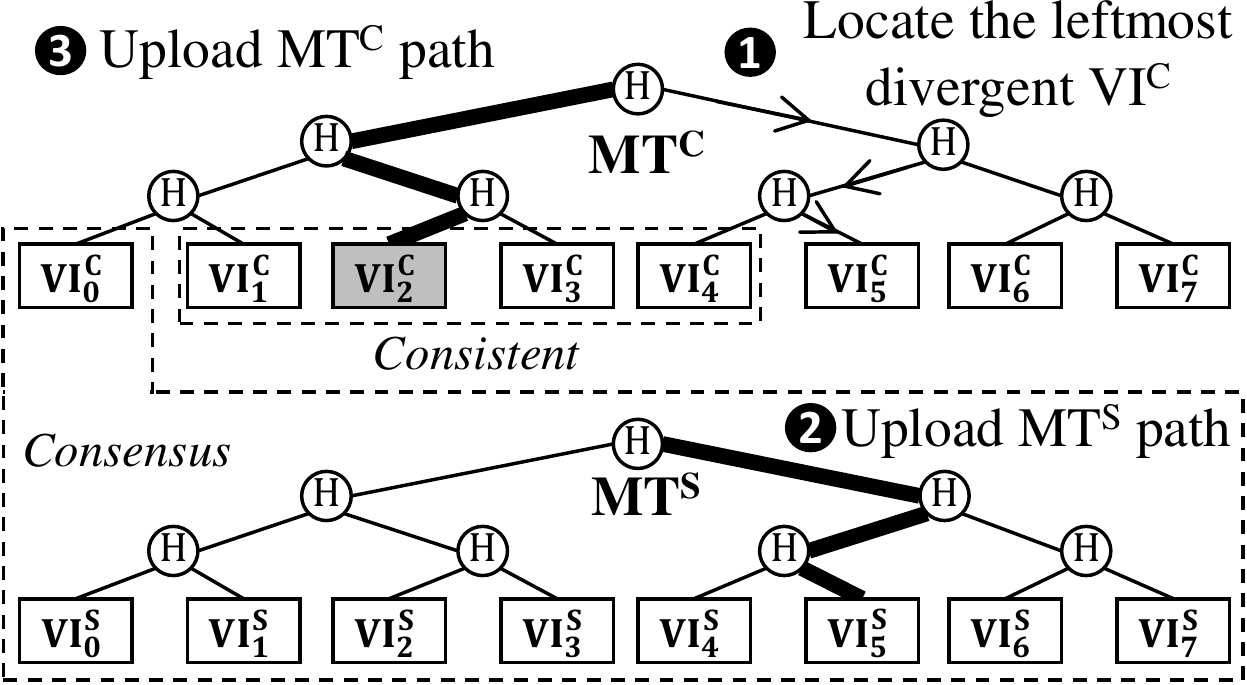}
  \vspace{-10pt}
  \caption{\protect\centering The MT-based commitment scheme to locate and arbitrate the divergent VI.} \label{fig:merkle}
  \vspace{-10pt}
\end{figure}

In the description above, for simplicity, the Merkle trees are shown as binary trees. In fact, they can be $k$-ary. Choosing the $k$ is a trade-off between the tree generation time, the number of interaction rounds and the gas consumption. In Agatha, we choose $k = 32$. The details will be provided in Section \ref{sec:eval}.

Up to this point, we have explained how to pinpoint a disputed BOP in a DNN computation, assuming the protocol can process the entire computation graph. However, this is difficult in reality. We will show in Section \ref{sec:eval} that the number of basic operations can be as large as 30 billion, the memory size to include these BOPs is over 600 GB. If we assume the number of OPs in the first-layer is $n_1$ and each OP has $n_2$ BOPs, the space complexity would be $O(n_1 \cdot n_2)$ either for generating or for evaluating the graph. To reduce the significant overhead, we propose two-phase pinpoint protocol, which introduces the Phase-1 evaluator.

\subsection{Introducing $E_{OP}$ for two-phase pinpoint} \label{sec:twophase}

To make pinpoint protocol practical, the pinpointing process needs to be done in two phases. Phase-1 identifies the disputed OP with another highly efficient evaluator $E_{OP}$. As shown in Figure \ref{fig:twophase}(a), the leaves of the Merkle tree for Phase-1 are OPs, so the size of the tree (i.e., P1\_MT$^C$) is much smaller. The procedure in Phase-1 is as what we described above, with two important details to note: (1) the VI represents an OP (e.g., (\textsf{MatMulInteger}, $\vec{in}$)) not a BOP (e.g., ``$\times$'', $in_1$, $in_5$). Variables $\vec{in}$ can be
tensors with arbitrary shape and data type.
The ``BOP'' field of VI in the previous section is replaced as one of operations (e.g., \textsf{Conv}) that takes input variables with certain dimensions. In our protocol, this OP is represented by the Merkle root of MT$^{S}$ of the operation's symbolic circuit (i.e., P2\_MT$^S$), rather than the string ``\textsf{MatMulInteger}'', so it is a well-defined concrete computation. (2) Different from evaluating a circuit, $E_{OP}$ leverages $E_N$ to evaluate every operation with native performance and computes based on those operation results, which is a feature of current DNN tool-chain \cite{onnxruntime}.

\begin{figure}[thb]
  \centering
  \includegraphics[width=0.46\textwidth]{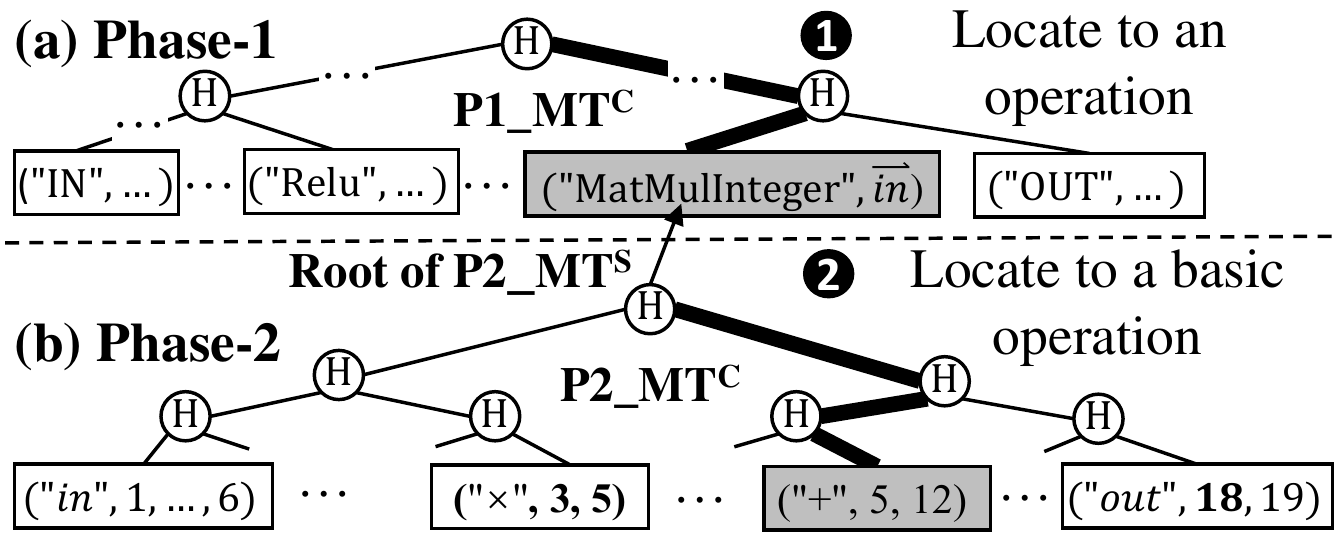}
  \vspace{-5pt}
  \caption{Two-phase pinpointing: (a) Phase-1 (b) Phase-2} \label{fig:twophase}
  \vspace{-10pt}
\end{figure}

Once the disputed operation is pinpointed in Phase-1, the submitter and the verifier start Phase-2, which is within the Merkle tree P2\_MT$^C$ in Figure \ref{fig:twophase}(b). Generating P2\_MT$^C$ and P2\_MT$^S$ is on-demand, so the space complexity of the protocol is reduced to $O(n_1 + n_2)$. In the end, a disputed BOP is submitted to the contract for arbitration.

\subsection{Contract arbitrator $E_{C}$}
We implement the contract arbitrator $E_C$ for all the BOPs in Table \ref{tab:bop} via the integer-based emulation, which is fully compliant with the current Ethereum contract standard. In other words, $E_C$ is aligned with $E_{BOP}$ with the granularity of BOP. The implementation carefully complies with IEEE-754 and demonstrates a negligible on-chain overhead for arbitration (see Section \ref{sec:eval}).

%% file: determine.tex
\section{Ensuring Consistency}\label{sec:determine}

Implementing the Agatha system requires several kinds of consistency as introduced in Section~\ref{sec:consistency}. In this section, we systematically explain the methods to eliminate all sources of inconsistency, from basic operations to DNN computation, so that we can ensure both cross-evaluator consistency and cross-platform consistency from bottom up.

\subsection{Basic operations ($E_C-E_{BOP}$ consistency)}
Both $E_{BOP}$ and $E_C$ evaluate basic operations. Among basic operations, integer operations are unambiguous; the main pitfall is due to floating-point operations.

Although the IEEE-754 standard strictly defines floating-point operations, it requires a thorough study to make $E_{BOP}$ and $E_C$ compliant to IEEE-754. For example, we must avoid implementing the minimum and maximum operations of two floating-point numbers by comparison (e.g., using C++ \texttt{std::max} in $E_{BOP}$), which violates the IEEE-754 standard to tackle corner cases such as signed zeros and NaNs. Non-standard mathematical functions, such as x86's \texttt{RCP} and \texttt{RSQRT}, should never be used.

\subsection{Operations ($E_{BOP}-E_{OP}$ consistency)}
Both $E_{OP}$ and $E_{BOP}$ evaluate operations.  $E_{BOP}$ is based on serialization of basic operations and requires the result consistent with $E_{OP}$. $E_{OP}$ evaluates an operation by invoking $E_N$, so it also requires cross-platform consistency. We have conducted a comprehensive study to make Agatha correct and efficient for evaluating operations.

\subsubsection{Inconsistent operations} \label{sec:source}
~

To identify cross-platform and cross-evaluator inconsistent operations, we have systematically reviewed the implementations of operations on the software/hardware platforms listed in Table \ref{tab:inconsistent} (covering different software, different processor micro-architectures, and various numbers of logical cores in the same processor), and investigated all the operations defined in ONNX \cite{onnx} operator set version 7 -- a total number of 100. We identify the inconsistent operations by analyzing their two essential sources:
\squishlisttwo
    \item Non-associativity of floating-point operations
    \item Standard-incompliant floating-point optimizations
    \squishtwoend

\para{Non-associativity.} Floating-point operations, as defined in the IEEE-754 standard \cite{IEEE754}, are not associative \cite{villa2009effects, goldberg1991every, monniaux2008pitfalls}. Thus, any computation that consists of multiple floating-point operations must specify the order of the operations, which is expressed as a computation graph, a.k.a. a \textit{circuit}. %
Otherwise, with different circuit representations, a computation that is consistent in real numbers becomes inconsistent in floating-point numbers.

\textit{Unfortunately, an operation only has a mathematical definition, which implies an ambiguous circuit representation.}  We observe that an operation may be interpreted into different circuit representations on different software/hardware platforms, for performance reasons. For example, Figure \ref{fig:cseq} illustrates three possible circuits for a \textsf{ReduceSum} operation that aggregates eight \texttt{float64} numbers: Circuit (a) shows the straightforward for-loop implementation; (b) shows the implementation using 128-bit SIMD instructions (e.g. SSE); (c) shows the implementation using 256-bit SIMD instructions (e.g. AVX). Besides the SIMD width, we notice that the circuit may also depend on the number of threads, the cache size, and the summation algorithm.

\begin{figure}[thpb]
  \centering
  \includegraphics[width=0.45\textwidth]{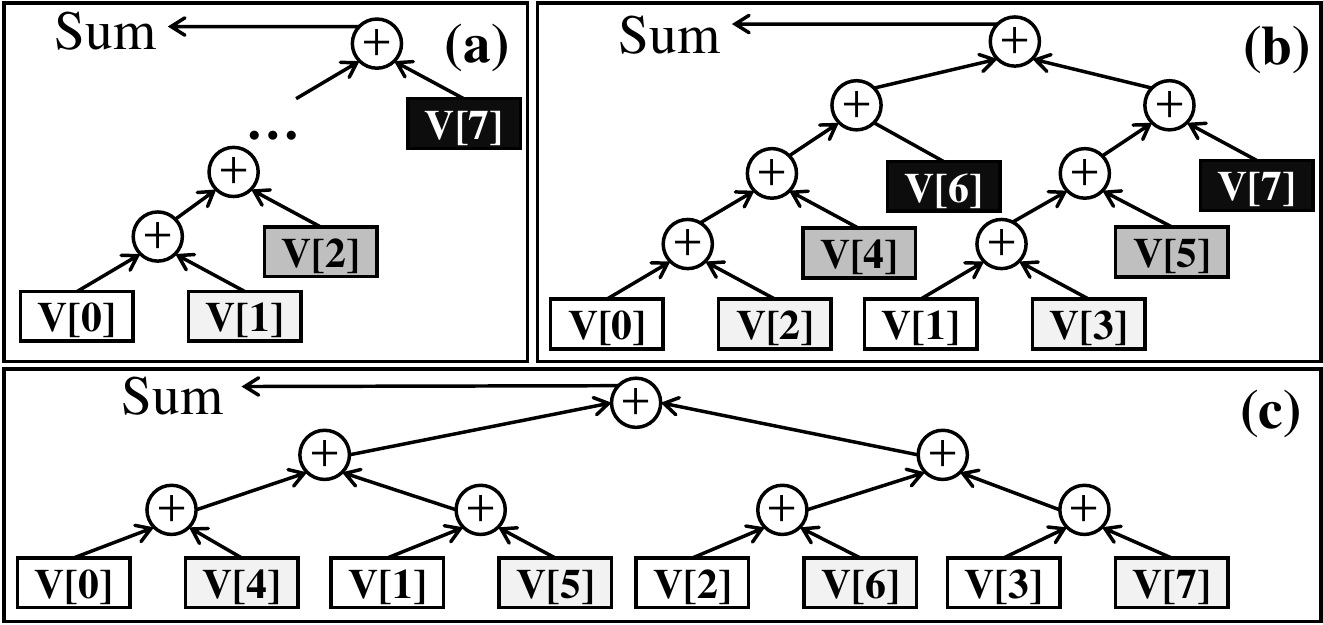}
  \vspace{-5pt}
  \caption{ The circuits with different summation orders: (a) sequential, (b) an order observed in an implantation in SSE (128-bit SIMD), (c) an order observed in AVX (256-bit SIMD). Elements of the same sub-figure with the same grey-scale are packed into one register.} \label{fig:cseq}
  \vspace{-5pt}
\end{figure}

\para{Standard-incompliant optimizations.}
An operation may be inconsistent even if it has an unambiguous circuit representation. We show that the inconsistency comes from three categories of standard-incompliant optimizations.

{\em First, using imprecise implementations.} Although the value of a mathematical function on floating-point numbers is well-defined according to the IEEE-754 standard,  some software enables approximate algorithms or specific hardware instructions to accelerate the mathematical computation at the cost of precision. For example, the Eigen library, which is used by TensorFlow and ONNX Runtime, is configured with ``\texttt{EIGEN\_FAST\_MATH=1}'' by default, so it uses imprecise implementations rather than its standard-compliant implementations. %

{\em Second, value-changing transformations.} Some software applies distributive laws and substitutes division with reciprocal multiplication in order to accelerate operations (such as \textsf{BatchNormalization}). %
However, this transforms the operation's existing circuit representation and may change the result~\cite{IEEE754}. %

{\em Third, improper compilation.} The computation can be non-standard due to improper compilation options. For example, ``\texttt{-{}-ffast-math}'' option in GCC \cite{gcc} enables unsafe compilation optimizations  and violate the IEEE standard even if the software avoids imprecise implementations or value-changing transformations. %

Based on our broad investigation about the operations and the platforms,
we summarize all the inconsistent operations in Table~\ref{tab:operators}. Remarkably, inconsistent operations due to non-associativity exist in almost every DNN computation. \textsf{Gemm} and \textsf{Conv}, shown in Figure \ref{fig:onnx}, are two such pervasively used operations. Thus, it is important to handle the inconsistency issue caused by floating-point non-associativity.

\begin{table}[t]
\centering
\footnotesize
\setlength{\tabcolsep}{2.0pt}
\renewcommand{\arraystretch}{1.1}
\begin{tabular}{cl}
\hline
\textbf{Essential source}  & \multicolumn{1}{c}{\textbf{Operations}} \\ \hline
Non-associativity         &  \begin{tabular}[c]{@{}l@{}}AveragePool, Conv, ConvTranspose, Gemm, \\ GlobalAveragePool,GlobalLpPool, GRU, LogSoftmax, \\ LpNormalization, LpPool, LRN, LSTM, MatMul, \\ Mean, ReduceL1, ReduceL2, ReduceLogSum, \\ ReduceLogSumExp, ReduceMean, ReduceProd, \\ ReduceSum, ReduceSumSquare, RNN, Softmax, Sum\end{tabular}   \\ \hline
\begin{tabular}[c]{@{}c@{}}Standard-incompliant\\ optimizations\end{tabular} &   \begin{tabular}[c]{@{}l@{}}\textbf{\textit{Square root}}: BatchNormalization, GlobalLpPool, \\ InstanceNormalization, LpNormalization, LpPool, \\ ReduceL2, Sqrt; \textbf{\textit{Max/Min}}: Clip, GlobalMaxPool, \\ HardSigmoid, Max, MaxPool, MaxRoiPool, Min, \\ ReduceMax, ReduceMin, Relu; \textbf{\textit{Transformation}}: \\ BatchNormalization, InstanceNormalization\end{tabular} \\ \hline
\end{tabular}
\caption{\protect\centering Inconsistent operations in ONNX Opset Version 7}
\label{tab:operators}
\vspace{-25pt}
\end{table}

\subsubsection{Addressing non-associativity for consistent operations}
~

Disabling optimizations incompliant with standard makes many operations consistent. The other inconsistent operations are due to the non-associativity of floating-point operations. These operations are usually the most computationally intensive part in a DNN, so we must consider performance factors when we analyze the following ways to to tackle these operations.

The most straightforward method is {\em sequencing}, i.e., assigning a particular circuit representation to an operation by applying the left associativity for every basic operation. With the unambiguous circuit, $E_{OP}$ and $E_{BOP}$ have the same behavior and naturally output the same result. However, sequencing loses the benefit of parallel execution and hardware SIMD instructions.

A variant of sequencing, {\em customized sequencing}, stipulates that the order of basic operations is the order running on a specific hardware platform, for both $E_{OP}$ and $E_{BOP}$. In this way, on the platform we can leverage the benefit of parallelism. For example, we can specify the order of summation for a \textsf{ReduceSum} operation to comply with an eight-core AVX-enabled CPU, like Figure~\ref{fig:cseq}(c). However, to guarantee cross-platform consistency, the computation on other hardware platforms must comply with the customized circuit. We consider this a difficult task for many reasons, such as the size of the circuit, participant's consensus and software implementation.

Another possibility is using data types that have the associative property. For example, ReproBLAS \cite{reproblas} is a linear algebra library in which the floating-point numbers are represented by ``binned types'' and the result of summation is reproducible and independent of the order of summation. ReproBLAS has the potential to eliminate the inconsistency from non-associativity. However, it is currently incomplete to support DNN operations and multi-thread execution, and too complicated for on-chain arbitration.

Fixed-point arithmetic is a simple and efficient solution. With the associative property of fixed-point arithmetic, we can customize the optimal computation order on different platforms while achieving consistency. $E_{OP}$ can perform parallel computing, so the performance improvement is significant compared to sequencing. Another advantage is that many existing libraries and hardware platforms support high-performance fixed-point arithmetic.

\begin{table}[]
\setlength{\tabcolsep}{6pt}
\renewcommand{\arraystretch}{1.00}
\small
\centering
\begin{tabular}{@{}cccc@{}}
\toprule
\textbf{} & Precision & Performance & Usability \\ \midrule
Sequencing & $\star$$\star$ & $\star$ & $\star$$\star$ \\
Customized sequencing & $\star$$\star$ & $\star$$\star$ & $\star$ \\
ReproBLAS & $\star$$\star$$\star$ & $\star$ & $\star$ \\
Fixed-point arithmetic & $\star$ & $\star$$\star$$\star$ & $\star$$\star$$\star$ \\ \bottomrule
\end{tabular}
\caption{\protect\centering Different methods to tackle floating-point non-associativity}
\label{tab:fp}
\vspace{-25pt}
\end{table}

The comparison of the above methods is shown in Table~\ref{tab:fp}. We adopt both sequencing and fixed-point arithmetic for Agatha. Specifically, in the context of DNN computation, using fixed-point arithmetic (or integer arithmetic\footnote{Fixed-point arithmetic is essentially based on integer arithmetic.}) is called ``quantization'', which has only an insignificant loss on the DNN model's accuracy ~\cite{jacob2018quantization,sze2017efficient}. In practice, if someone wants the insignificant accuracy loss even smaller (e.g., from $1\%$ to $0.1\%$), one can use partial quantization. For example, most matrix multiplication-based operations (such as \textsf{Gemm}, \textsf{Conv}, \textsf{RNN}, \textsf{LSTM}) are quantized, but others still take floating-point numbers~\cite{onnxruntime}. Considering this accuracy-performance tradeoff, Agatha adopts fixed-point arithmetic for quantization operation and sequences others, faithfully following the DNN model description.

\subsection{Graphs ($E_{OP}-E_N$ consistency)}
Both $E_N$ and $E_{OP}$ evaluate the whole computational graph of a DNN. We have ensured that the operations are consistent, so it is easy to ensure the whole graph's consistency. The main pitfall here is the graph optimization in $E_N$. Some software's default graph optimization level (such as ONNX Runtime's) makes graph optimization so deep that the graph transformation changes the floating-point values. Setting a safer level resolves the issue.

%% file: security.tex
\section{Security Analysis}\label{sec:security}

In this section, we provide the security analysis of our Agatha system with formal definitions and proofs.
\subsection{Definitions} \label{sec:def}
The security definitions in this paper are very similar to those of Arbitrum \cite{arbitrum}. The main difference is that we introduce the native DNN computation in the AnyTrust assumption and use GPP instead of VM-based pinpoint protocol. Besides, we also introduce the definitions for the evaluators used in Agatha system:
\begin{definition}[AnyTrust assumption.] \label{def:anytrust}
There are a quorum of off-chain nodes \set{N_1,\ldots,N_m} with their corresponding platforms \set{P_1,\ldots,P_m} for DNN computation, smart contract $\mathcal{C}$ as the arbitrator, a native DNN computation $\phi^n$, native evaluator $E_N$, input data $\mathcal{D}$, and a graph-based pinpoint protocol $\Pi$. We use $N_s\in \set{N_1,\ldots,N_m}$ to denote the submitter.
AnyTrust assumption holds if for every claim of $N_s$ that ``$\mathcal{R}_s$ = $E_N$(P$_s$, $\mathcal{D}$, $\phi^n$)'' sent to other nodes in \set{N_1,\ldots,N_m} with corresponding commitment to contract $\mathcal{C}$, either $N_s$ {\em honestly participates} in $\Pi$ or there exists a verifier $N_v\in \set{N_1,\ldots,N_m}$ ($v \neq s$) to correctly compute $\mathcal{R}_v$= $E_N$(P$_v$, $\mathcal{D}$, $\phi^n$) and to {\em honestly participate} in $\Pi$ if $\mathcal{R}_v\neq \mathcal{R}_s$.
\end{definition}

\begin{definition}[Evaluators in Agatha system.] \label{def:dnn}
There are four evaluators ($E_N$, $E_{OP}$, $E_{BOP}$ and $E_C$):
\squishlist
\item {\em For native execution,} $E_N$ is the evaluator for native DNN computation $\phi^n$ as the mapping: $P_{*}$ $\times \mathcal{D} \times$ $\phi^n$  $\mapsto$ $\mathcal{R}_{*}$, where $P_{*}$ denotes the platform to run $E_N$.
\item {\em For Phase-1 pinpoint protocol,} $\phi^n$ is conventionally expressed as a sequence of operations (OPs): \set{\phi^{op}_1,\phi^{op}_2,\ldots,\phi^{op}_k}. $\phi^{op}_i$ is executed by $E_{OP}$ as the mapping: $P_{*}$ $\times$ $\mathcal{D}$$^{op}_i$ $\times$ $\phi^{op}_i$ $\mapsto$ $\mathcal{R}^{op}_i$, where $\mathcal{D}$$^{op}_i$ and $\mathcal{R}^{op}_i$ denote the input and output of $\phi^{op}_i$, respectively. D$^{op}_i$ comes from either input $\mathcal{D}$ or previous output $\mathcal{R}^{op}_j$ ($j<i$). $E_{OP}$ can also calculate the intermediate result for OPs: $P_{*}$ $\times$ $\mathcal{D}$ $\times$ \set{\phi^{op}_1,\phi^{op}_{2},...,\phi^{op}_i} $\mapsto$ $\mathcal{R}^{op}_{i}$.
\item {\em For Phase-2 pinpoint protocol,} $\phi^{op}_i$ can be further serialized as a sequence of basic operations (BOPs): \set{\phi^{bop}_{i,1},\phi^{bop}_{i,2},\ldots,\phi^{bop}_{i,l_i}} (e.g., by topological sorting).  $E_{BOP}$ is the mapping: $P_{*}$ $\times \mathcal{D}^{bop}_{i,j}$ $\times$ $\phi^{bop}_{i,j}$ $\mapsto$ $\mathcal{R}^{bop}_{i,j}$, where $\mathcal{D}$$^{bop}_{i,j}$ and $\mathcal{R}^{op}_{i,j}$ denote the input and output of $\phi^{bop}_{i,j}$, respectively. Similar to $E_{OP}$, $E_{BOP}$ can also evaluate the intermediate result: $P_{*}$ $\times \mathcal{D}^{op}_i$ $\times$ \set{\phi^{bop}_{i,1},\phi^{bop}_{i,2},\ldots,\phi^{bop}_{i,j}} $\mapsto$ $\mathcal{R}^{bop}_{i,j}$.

\item {\em For contract arbitration,} $E_C$ evaluate contract arbitration $\phi^c$ as the mapping: $P^E_{*}$ $\times$ $\mathcal{D}$$^{bop}$ $\times$ $\phi^c$ $\mapsto$ $\mathcal{R}^{bop}$, where $P^E_{*}$ denotes the platform to run Ethereum client while $\mathcal{D}$$^{bop}$ and $\mathcal{R}^{bop}$ denote the input and output for $\phi^c$, respectively. Obviously, $E_C$ is cross-platform consistent, namely $E_C$($P^E_{*}, \mathcal{D}^{bop}$, $\phi^c$) = $E_C$($P'^E_{*}, \mathcal{D}^{bop}$, $\phi^c$) for arbitrary platforms $P^E_{*}$ and $P'^E_{*}$ with Ethereum clients.
\squishend
\end{definition}

\begin{definition}[Graph-based pinpoint protocol (GPP) $\Pi$.] \label{def:protocol}
$\Pi$ is encoded in contract $\mathcal{C}$ and triggered by the claim of submitter $N_s$. $\Pi$ provides two exclusive APIs \set{\pi^S, \pi^V} for the challenge of verifier $N_v$ $(v\neq s)$ and the response of submitter $N_s$, respectively. Besides the Merkle tree verification logic, four parameters are pre-determined: (1) a timeout period T$^v$ = O(|$\phi^n$|) which is proportional to the latency of $E_N$(P$_v, \mathcal{D}, \phi^n$)) and specifies the maximum interval for the first verifier to invoke $\pi^V$; (2) a timeout period T$^{op}$ = O(|$\phi^{op}$|) which is proportional to the latency of $E_{OP}$(P$_v$, $\mathcal{D}$,\set{\phi^{op}_1,\phi^{op}_2,\ldots,\phi^{op}_k}) for Phase-1 protocol and specifies the maximum interval between the verifier-submitter interactions (via $\pi^S$ and $\pi^V$) during the pinpoint protocol; (3) T$^{bop}$ = O(|$\phi^{bop}$|) for evaluating Phase-2 protocol which is similar to T$^{op}$; (4) the termination conditions (either for the aborting case or the non-aborting case) for the pinpoint protocol, which includes the contract arbitrator ${E_C}$ for $\phi^c$ (i.e., BOP).

 $N_s$ and $N_v$ are considered {\em honestly participating} in $\Pi$ if $N_s$ and $N_v$ always correctly interact with $\mathcal{C}$ within the time-out periods (i.e., T$^v$, T$^{op}$ and T$^{bop}$) according to contract states and their evaluation results of $\phi^n$, $\phi^{op}$, $\phi^{bop}$ with $E_N$, $E_{OP}$, $E_{BOP}$, respectively. Once $\pi^V$ is invoked, we assume that the count down of T$^v$ is suspended until termination. Verifier $N_v$ {\em fails in the pinpoint protocol} if either the time-out is reached or contract arbitration indicates $N_s$ is right. Failed verifier can no longer invoke $\pi^V$ for this claim. If all verifiers fail or there is no challenge within T$^v$, the contract {\em accepts} $N_s$'s claim otherwise {\em rejects} the claim.
\end{definition}

Following ACE \cite{ace}, the core security properties of Agatha system is the correctness and the liveness:

\begin{definition}[Correctness of Agatha.] \label{def:correctness}
Agatha's correctness is twofold: (1) The contract will not accept $N_s$'s wrong claim given the AnyTrust assumption and our design of GPP and XCE; (2) $N_s$'s correct claim will not be challenged by honest verifiers and will be accepted by the contract given malicious verifiers.
\end{definition}

\begin{definition}[Liveness of Agatha.] \label{def:liveness}
$N_s$'s claim will be either accepted or rejected by the contract within a maximum period $T_{max}$ even if $N_s$ and $m-1$ verifiers maliciously invoke $\pi^S$ and $\pi^V$, respectively.
\end{definition}

XCE plays the key role to guarantee the correctness of GPP and the whole Agatha system. In particular, XCE guarantees the cross-evaluator consistency between $E_N$ - $E_{OP}$, $E_{OP}$ - $E_{BOP}$ and $E_{BOP}$ - $E_{C}$, which can further ensure the cross-platform consistency:

\begin{definition}[Cross-evaluator consistency of XCE.] \label{def:consistency}
The consistency between three pairs of evaluators (i.e., $E_N$ - $E_{OP}$, $E_{OP}$ - $E_{BOP}$ and $E_{BOP}$ - $E_{C}$) are defined as follows:
\squishlist
 \item {$E_N$ - $E_{OP}$,} which means $\forall$ $\mathcal{D}$ and $\forall$ P$_* \in \set{P_1,\ldots,P_m}$, the output of $E_{OP}$ for the OP sequence is the same with $E_{N}$'s output $\mathcal{R}_{*}$:
 $E_{OP}$(P$_*$, $\mathcal{D}$,\set{\phi^{op}_1,\phi^{op}_2,\ldots,\phi^{op}_k}) = $E_N$(P$_*$, $\mathcal{D}$, $\phi^n$).
\item {$E_{OP}$ - $E_{BOP}$,} which means $\forall \mathcal{D}^{op}_i$ and $\forall$ P$_* \in \set{P_1,\ldots,P_m}$, the output of $E_{BOP}$ for the BOP sequence and $E_{OP}$'s output are same:
 $E_{BOP}$(P$_*, \mathcal{D}^{op}_i$,\set{\phi^{bop}_{i,1},\phi^{bop}_{i,2},\ldots,\phi^{bop}_{i,l_i}} ) = $E_{OP}$(P$_*, \mathcal{D}^{op}_i$, $\phi^{op}_i$).
 \item {$E_{BOP}$ - $E_{C}$,} which means $\forall \mathcal{D}^{bop}_{i,j}$, $\forall$ P$_* \in \set{P_1,\ldots,P_m}$ and $\forall$ $P^E_{*}$ running Ethereum, the output of $E_{BOP}$ for $\mathcal{D}$$^{bop}_{i,j}$ is the same with $E_{C}$'s output:
 $E_{BOP}$(P$_*, \mathcal{D}^{bop}_{i,j}$,$\phi^{bop}_{i,j}$) = $E_{C}$($P^E_{*}$, $\mathcal{D}$$^{bop}_{i,j}$, $\phi^{c}$).
\squishend
\end{definition}

\begin{definition}[Cross-platform consistency.] \label{def:platform}
$\forall$ $\mathcal{D}$ and $\forall$ P$_i$, P$_j$ $\in \set{P_1,\ldots,P_m}$, their outputs of $E_{N}$, $E_{OP}$, $E_{BOP}$ are equal: (1) $E_N$(P$_i, \mathcal{D}, \phi^n$) = $E_N$(P$_j, \mathcal{D}, \phi^n$); (2)  $E_{OP}$(P$_i$, $\mathcal{D}$,\set{\phi^{op}_1,\phi^{op}_2,\ldots,\phi^{op}_k}) = $E_{OP}$(P$_j$, $\mathcal{D}$,\newline$\{\phi^{op}_1, \phi^{op}_2,$ $\ldots,\phi^{op}_k\}$) ; (3) $E_{BOP}$(P$_i$, $\mathcal{D}$$^{op}_i$,\set{\phi^{bop}_{i,1},\phi^{bop}_{i,2},\ldots,\phi^{bop}_{i,l_i}} ) = $E_{BOP}$ (P$_j$, $\mathcal{D}$$^{op}_i$, \set{\phi^{bop}_{i,1},\phi^{bop}_{i,2},\ldots,\phi^{bop}_{i,l_i}} ).
\end{definition}

\subsection{Security proof} \label{sec:proof}
\begin{lemma}{XCE guarantees the cross-platform consistency.}\label{th:2}
\end{lemma}
\begin{proof}
$\forall$ $\mathcal{D}$ and $\forall$ P$_i$, P$_j$ $\in \set{P_1,\ldots,P_m}$, with the cross-platform consistency of $E_C$ and the cross-evaluator consistency of $E_C$-$E_{BOP}$, we can get  $E_{BOP}$(P$_i$, $\mathcal{D}$$^{bop}$,$\phi^{bop}$) = $E_{BOP}$(P$_j$, $\mathcal{D}$$^{bop}$,$\phi^{bop}$), namely $E_{BOP}$ is cross-platform consistent. With the $E_{OP}$ - $E_{BOP}$ and $E_{OP}$ - $E_{N}$ consistency, we can further conclude $E_{OP}$(P$_i$, $\mathcal{D}$$^{op}_i$, $\phi^{op}_i$)=$E_{OP}$(P$_j$, $\mathcal{D}$$^{op}_i$, $\phi^{op}_i$) and $E_N$(P$_i$, $\mathcal{D}$, $\phi^n$) = $E_N$(P$_j$, $\mathcal{D}$, $\phi^n$).
\end{proof}

\begin{theorem}[The Correctness of Agatha]\label{th:1}
{ The claim of $N_s$ is accepted by contract $\mathcal{C}$ without challenges from honest verifiers if and only if $\mathcal{R}_s$ = $E_N$(P$_*$, $\mathcal{D}$, $\phi^n$).}
\end{theorem}
\begin{proof}
 If $\mathcal{R}_s$ = $E_N$(P$_*$, $\mathcal{D}$, $\phi^n$), according to Lemma \ref{th:2} and Anytrust assumption, the honest verifiers will not challenge the claim since $\mathcal{R}_v = \mathcal{R}_s$.

 If $\mathcal{R}_s$ $\neq$ $E_N$(P$_*$, $\mathcal{D}$, $\phi^n$), we can prove that GPP will pinpoint to a $\phi^{bop}$ (unless $N_s$ aborts midway), where $\mathcal{R}^{bop}_s$ $\neq$ $E_{BOP}$(P$_*$, $\mathcal{D}$$^{bop}$, $\phi^{bop}$). This inequality will be arbitrated by E$_C$ and lead to the rejection of $N_s$'s claim. If we assume that there is no such $\phi^{bop}$, according to the $E_{OP}$ - $E_{BOP}$ and $E_{OP}$ - $E_{N}$ consistency, we can have $E_N$(P$_s$, $\mathcal{D}$, $\phi^n$) = $E_{BOP}$(P$_*$, $\mathcal{D}$$^{op}_i$, \set{\phi^{bop}_{i,j}}), where \set{\phi^{bop}_{i,j}} denotes the BOP sequence for \set{\phi^{op}_1,\phi^{op}_2,\ldots,\phi^{op}_k}. Because $\mathcal{R}^{bop}_s$ $=$ $E_{BOP}$(P$_*$, $\mathcal{D}$$^{bop}$,$\phi^{bop}$) for every $\phi^{bop}_{i,j}$, we have $E_N$(P$_s$, D, $\phi^n$)$\triangleq$ $\mathcal{R}_s$ = $E_N$(P$_*$, $\mathcal{D}$, $\phi^n$), contradicting to $\mathcal{R}_s$ $\neq$ $E_N$(P$_*$, $\mathcal{D}$, $\phi^n$).

 Similarly, we can also get that if $\mathcal{R}_s$ = $E_N$(P$_*$, $\mathcal{D}$, $\phi^n$), GPP will find $\mathcal{R}^{bop}_v$ $\neq$ $E_{BOP}$(P$_*$, $\mathcal{D}$$^{bop}$,$\phi^{bop}$) for the malicious verifiers' challenge. Therefore, all malicious verifiers will fail.
\end{proof}

\begin{theorem}[The Liveness of Agatha]\label{th:4}
Within maximum period {T$_{max}$ = T$^v$ + $m\times$(d($\vert\phi^{op}\vert$) $\times$T$^{op}$+ d($\vert\phi^{bop}\vert$) $\times$T$^{bop}$), where d($\vert\phi^{*}\vert$) is the depth of the k-ary Merkle tree representing $\phi^{*}$, the claim of $N_s$ will be either accepted or rejected.}
\end{theorem}
\begin{proof}
 If $N_s$'s claim is true and verifiers are all honest, the claim will be accepted within T$^v$. The worst case is that m-1 verifiers maliciously invoke $\pi^V$ at T$^v$. For every challenge from malicious verifiers, the process takes at most d($\vert\phi^{op}\vert$) $\times$T$^{op}$+ d($\vert\phi^{bop}\vert$) $\times$T$^{bop}$, which corresponds to T$_{max}$ = T$^v$ + $(m-1)\times$(d($\vert\phi^{op}\vert$) $\times$T$^{op}$+ d($\vert\phi^{bop}\vert$) $\times$T$^{bop}$).

 If $N_s$'s claim is false, normally it will be rejected with honest verifier's challenge in d($\vert\phi^{op}\vert$) $\times$T$^{op}$+ d($\vert\phi^{bop}\vert$) $\times$T$^{bop}$. There exists a worst case that the honest verifier and other malicious verifiers invoke $\pi^V$ at T$^v$, which corresponds to the rejection of $N_s$'s claim in T$_{max}$=T$^v$ + $m\times$(d($\vert\phi^{op}\vert$) $\times$T$^{op}$+ d($\vert\phi^{bop}\vert$) $\times$T$^{bop}$).
 \end{proof}

%% file: eval.tex
\section{Evaluations}\label{sec:eval}

In this section, we provide the details of our implementation and experiment setup. Then, we show the evaluations of Agatha's correctness and performance.

\subsection{Implementation and experiment setup}
\para{Implementation.} As shown in Figure \ref{fig:overview}, we implement the Agatha system with five steps: (1) With methods in Section \ref{sec:determine}, we implement XCE based on the ONNX Runtime \cite{onnxruntime}, targeting three models (i.e., MobileNet, ResNet50, VGG16) and use XCE-enabled native evaluator $E_{N}$.
(2) We integrate the rule-based generating and the topological serialization into our circuit generator with $\sim$1700 lines of C++ code. We speed up the generator by reusing duplicate circuits, given the recurrence of operations.
(3) We implement the logic of Merkle tree, the interface to the contract, Phase-1 and Phase-2 evaluators $E_{OP}$ and $E_{BOP}$, and a command-line interface in the Agatha client, which can connect to Ethereum and run the pinpoint protocol automatically for the submitter and the verifier. We implement the Agatha client with $\sim$4000 lines of Python code (testing logic included), which leverages the ONNX library \cite{onnx} and Brownie framework \cite{brownie}. The module for Merkle tree generating is highly optimized in C++ leveraging the XKCP library \cite{XKCP}, which provides the Keccak-256 hashing algorithm.
The leaves of the Merkle tree are fixed-length to enable packing before hashing. To balance the Merkle tree, we move the VI$^C$ and VI$^S$ of out/in to be under the Merkle root. We further compress the field of inputs in each VI with a Merkle tree structure. (4) We implement the Agatha contract in Solidity 0.7.1 \cite{solidity} with $\sim$500 lines of code, which specifies the logic for submitter-verifier interactions, Merkle path verification and the arbitration of BOP. The gas consumption is optimized in several ways,
such as leveraging gas refunds \cite{EIP2200} and using inline assembly in Solidity.
(5) We implement the BOP contract library with $\sim$1400 lines of Solidity code, in which the 32-bit floating-point arithmetic is adapted from the ABDK library \cite{abdk}.

\para{Setup.} For the off-chain nodes, we use Azure machines with 24 vCPUs (E5-2690 v3), which is the configuration for Section \ref{sec:perf}. We also use various machines of Table \ref{tab:hardware} to test cross-evaluator consistency defined in Section \ref{sec:security}.
For the nodes of Ethereum, we use Ganache \cite{ganache} to simulate the public Ethereum (Istanbul version \cite{istanbul}), which reports the same gas consumption.
And for benchmarks, we use the standard test images from ImageNet \cite{imagenet} and ONNX-format files \cite{onnxzoo} of three quantized models (i.e., MobileNet, ResNet50, VGG16) following \cite{slalom}. For the Merkle tree configuration, we set the branch size as 32 by default since it is optimal for the gas consumption (detailed in Figure \ref{fig:gas_tradeoff}).

\subsection{Validation about consistency} \label{sec:correct}
As shown in Section \ref{sec:security}, the correctness of Agatha system relies on the cross-evaluator consistency which are introduced in Section \ref{sec:determine} and defined in Section \ref{sec:security}. Besides our careful design and investigation, we also have conducted comprehensive tests to evaluate the cross-evaluator consistency. All cases are tested on the different hardware of Table \ref{tab:inconsistent}(b):
(1) {For $E_{BOP}$-$E_C$ consistency,} we use random inputs and $\sim$7500 floating-point corner cases as the unit tests for the $E_{OP}$-$E_{BOP}$ consistency. The subnormal numbers, infinities, signed zeros and NaNs are all included to show consistency;
(2) {For $E_{OP}$-$E_{BOP}$ consistency,} we test more than 20 selected corner cases and more than 1000 random cases. All tests demonstrate consistency;
(3)  {For $E_{N}$-$E_{OP}$ consistency,} the $E_{OP}$ and $E_{N}$ demonstrate consistent output hashes with all test images from ImageNet \cite{imagenet}.

\subsection{Performance evaluation}  \label{sec:perf}
We first profile the three models used in this work. As shown in Table \ref{tab:profile}, we provide the storage size, the number of OPs and the number of BOPs in each model, respectively. We see that the model can be as large as 133 MB with 30 billion BOPs. In addition, the recurrence of OPs is common, so our deduplication (i.e., circuit reuse) can reduce the latency of the circuit generation substantially. Due to the page limit, we provide details of OPs (e.g., OP size distribution, circuit generation) in Appendix \ref{sec:understand}.

\begin{table}[thbp]
\centering
\footnotesize
\setlength{\tabcolsep}{8pt}
\renewcommand{\arraystretch}{1.2}
\begin{tabular}{cccccc}
\hline
Model     & Storage   & \multicolumn{2}{c}{Before dedup} & \multicolumn{2}{c}{After dedup} \\ \hline
{}        &        & \#OP    & \#BOP                  & \#OP   & \#BOP                  \\ \cline{3-6}
MobileNet & 3.8 MB & 371     & $9.9 \times 10^8$      & 131    & $5.5 \times 10^8$      \\
ResNet50 & 25 MB  & 387     & $8.1 \times 10^9$      & 97     & $3.9 \times 10^9$      \\
VGG16     & 133 MB & 186     & $3.0 \times 10^{10}$   & 75     & $2.2 \times 10^{10}$   \\ \hline
\end{tabular}
\caption{The number of OPs and BOPs in each model} \label{tab:profile}
\vspace{-15pt}
\end{table}

\para{The performance of Agatha evaluators.}
 Figure \ref{fig:normallatency} shows the latency of the native evaluator $E_N$ for running each of the three models end-to-end, which is the average result of 10,000 tests. The latency of ours includes the inference time and the time to generate the output hash. Compared with the original evaluator (i.e., ONNX runtime), $E_N$ introduces latency overheads as 6.6\%, 1.4\% and 1.3\% for MobileNet, ResNet50 and VGG16 (the geometric mean is 3.0\%) , respectively. The overhead indicates the power of our design and the limited impact of turning off ``\texttt{-{}-ffast-math}'' option.

 We also test and estimate the latency for ideal VM and restricted VM, which denote the single-threaded native execution (the upper-bound performance of VM-based approaches) and the restricted VM of previous work, respectively. And for simplicity, we use Arbitrum VM as an example of the restricted VM which is based on EVM and actively being developed. Since EVM does not support floating-point arithmetic and has a gas limit, we use integer arithmetic for approximation and measure the speed of different instructions \cite{perez2019broken} to estimate EVM's lower-bound latency. We can see that $E_N$ demonstrates an average speedup of 2.59$\times$ and 602$\times$ compared to ideal VM and restricted VM, respectively.

\begin{figure}[thpb]
  \centering
    \vspace{5pt}
  \includegraphics[width=0.45\textwidth]{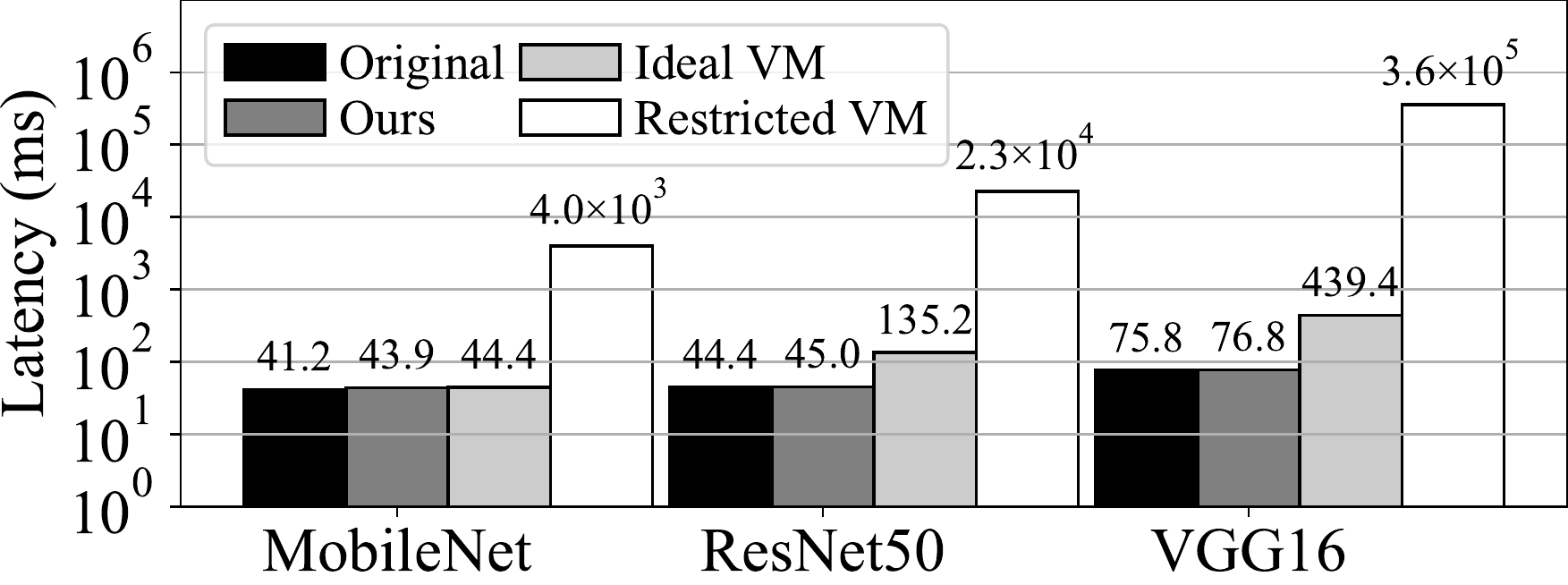}
  \vspace{-5pt}
  \caption{\protect\centering Inference latency of the original (i.e., baseline), our native evaluator, the ideal VM and the restricted VM}
  \vspace{-5pt}
   \label{fig:normallatency}
\end{figure}

Figure \ref{fig:p2latency} shows the local execution latency of Phase-1 and Phase-2 evaluators $E_{OP}$ and $E_{BOP}$, respectively. We can see that the latency of $E_{OP}$ is relatively small since we leverage the native execution of ONNX Runtime for evaluating every operation. To evaluate the $E_{BOP}$ latency, for simplicity, we only showcase the largest operation in each model (e.g., \textsf{ConvInteger} in VGG16) which represents the worst-case latency. The worst-case Phase-2 latency depends on the operation size and $E_{BOP}$ performance, which can be smaller or larger than Phase-1's. In addition, we also show the latency for one-phase pinpoint protocol where the whole DNN computation is expressed and evaluated in BOPs. We can find that the two-phase pinpoint greatly reduces the latency of the pinpoint protocol.
\begin{figure}[thb]
  \centering
    \vspace{-5pt}
  \includegraphics[width=0.42\textwidth]{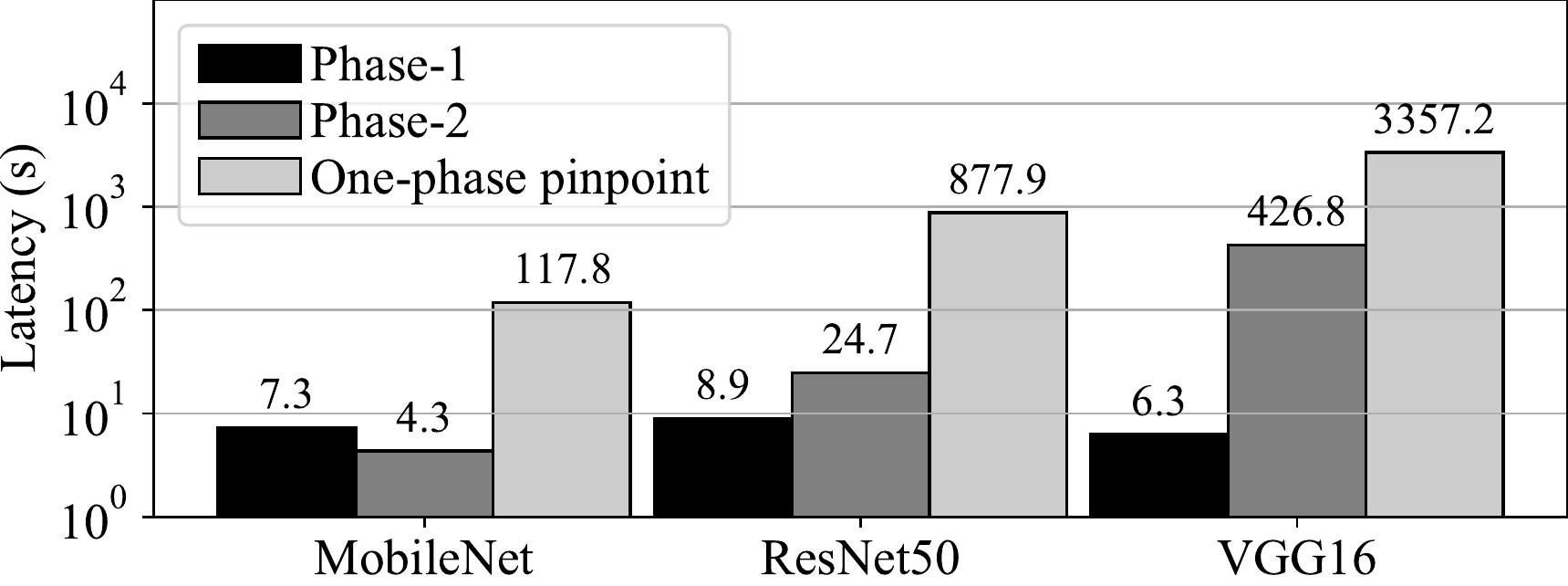}
  \vspace{-5pt}
  \caption{\protect\centering The Phase-1 and Phase-2 latency vs. latency of One-phase pinpoint}
  \vspace{-5pt}
   \label{fig:p2latency}
\end{figure}

\para{Gas consumptions.} In Figure \ref{fig:gas_tradeoff}(a), we show the gas consumptions of the main functions in our contract. Since we have applied the inline mechanism, the gas consumptions are calculated by differential analysis. For the setup, the Agatha contract and the BOP library are deployed to Ethereum with $\sim 4.8 \times 10^6$ gas, which is one-time. For the normal case, the Submit (i.e., submitter's on-chain commitment) costs $\sim$63,000 gas which corresponds to three ETH transfer transactions. For the pinpointing cost, we see that the gas consumed by BOPs is negligible compared to those of other functions. The maximal gas consumption for the whole pinpointing (i.e., Max total of VGG16) is about $\sim$86 ETH transfer transactions.

 \begin{figure}[thb]
   \centering
         \begin{subfigure}[t]{.23\textwidth}
       \centering
       \makebox(0,15)[l]{\hspace{0.3em}\textbf{(a)}}%
       \includegraphics[width=\textwidth]{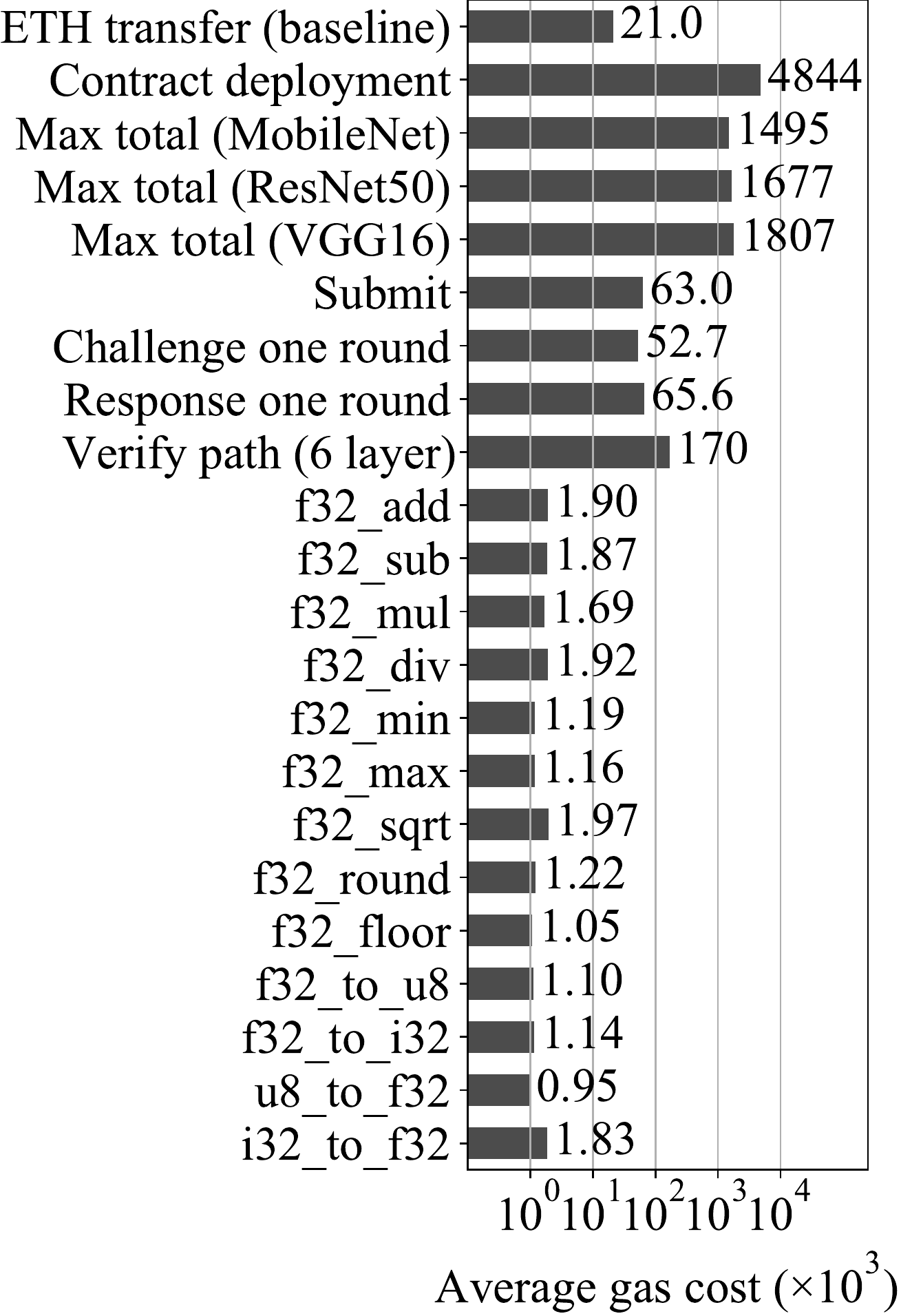}
     \end{subfigure}
     \begin{subfigure}[t]{.23\textwidth}
       \centering
       \makebox(0,15)[l]{\hspace{0.3em}\textbf{(b)}}%
       \includegraphics[width=\textwidth]{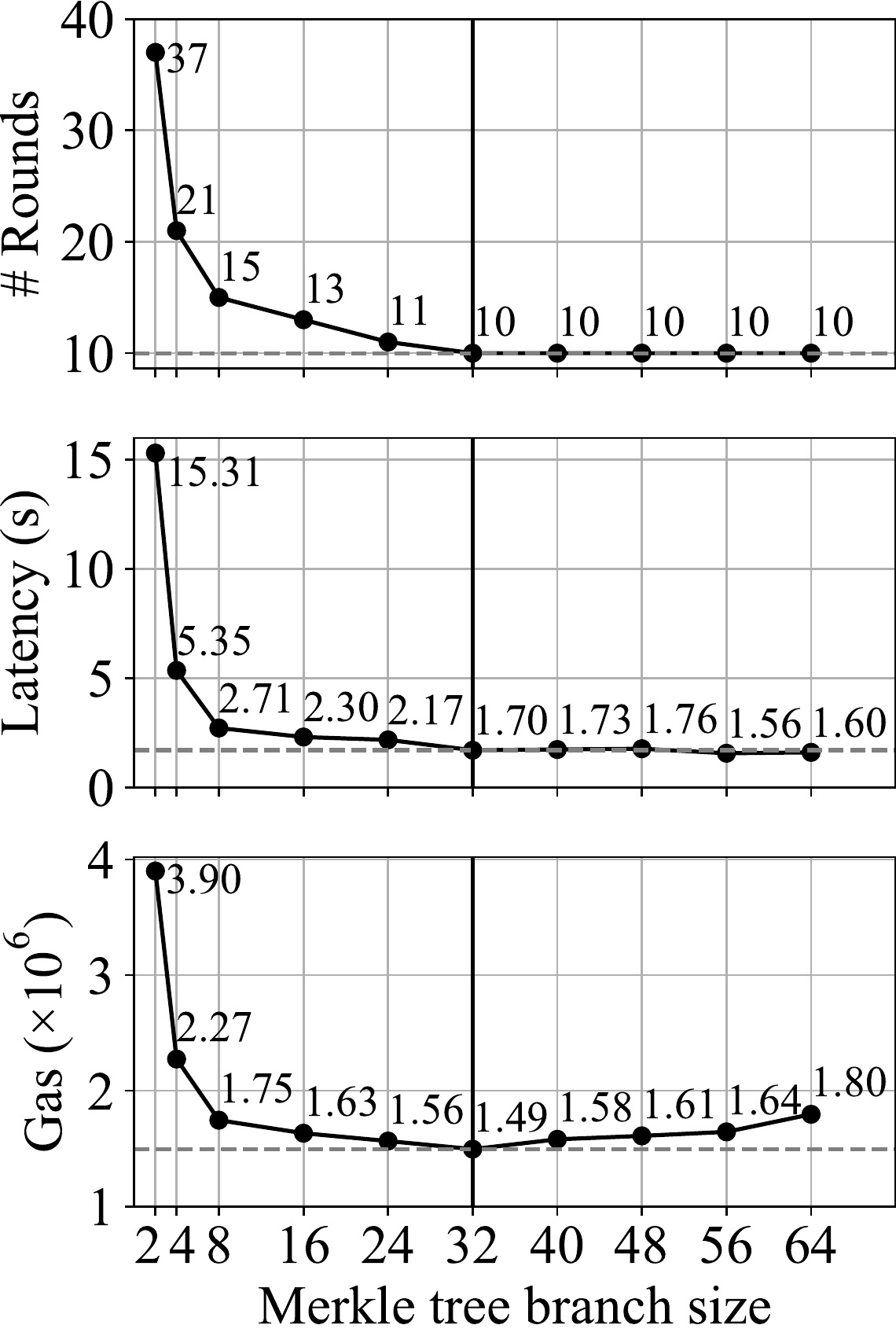}
     \end{subfigure}\hfill%
     \vspace{-5pt}
     \caption{\protect\centering (a) The breakdown of gas consumption (b) The optimal choice of Merkle tree branch size for MobileNet}\label{fig:gas_tradeoff}
     \vspace{-5pt}
 \end{figure}

\para{Merkle tree branches.} We also study the impact of the number of the Merkle tree branches on the performance, as shown in Figure \ref{fig:gas_tradeoff}(b). Due to the space constraint, we only show the case of the MobileNet, which is similar to those of ResNet50 and VGG16. When the branch size varies from 2 to 64, we see a decrease in the number of interaction rounds and the Merkle tree generating time, because the depth of Merkle tree decreases. Interestingly, the gas consumption for pinpointing is minimal when the branch size is 32, which is caused by two effects: (1) The number of rounds decreases with a larger branch size, which reduces the total gas of Challenge and Response in Figure \ref{fig:gas_tradeoff}(a). (2) The increase of branch size leads to more hashing workloads, which introduces the high overhead of Verify-Path in Figure \ref{fig:gas_tradeoff}(a).

%% file: relate.tex
\section{Related Work}\label{sec:relate}

\para{AI computation on blockchain.} The vision of AI computation on blockchain is a subject mentioned by researchers and companies. In academia, Das {\em et al.} mention machine learning as a potential area for ``Computationally Intensive Smart Contracts'' \cite{yoda}. And Teutsch {\em et al.} regard ``Autonomous machine learning'' as a promising application for TrueBit \cite{truebit} while Wüst {\em et al.} figure out that machine learning is infeasible for current smart contracts \cite{ace}. Other researchers discuss the potential of enabling machine learning on a blockchain for smart home \cite{smarthome} and cooperative data-driven applications \cite{cdda}. In industry, Microsoft proposes ``Decentralized \& Collaborative AI'' on blockchains \cite{microsoft}. Several startups, including Cortex \cite{cortex}, SingularityNET \cite{singularitynet}, Algorithmia \cite{algorithmia} and Oraclize \cite{oraclize}, share this vision. However, the existing proposals either demonstrate tiny AI computations that achieve the consensus on a handful of nodes \cite{cdda,smarthome,microsoft,algorithmia,singularitynet,cortex,oraclize} or provide no details \cite{yoda,truebit,ace}. They are not scalable either for both on-chain or off-chain execution. It is impractical for Ethereum smart contract to use these approaches to run a real-world DNN computation.

\para{Scalable smart contract.}
Due to the low scalability of blockchain \cite{statechannel,payment}, there are currently two mainstream approaches to offload the computation from smart contract to off-chain nodes: either via the {\em validity proof} or via the {\em fraud proof}. The first approach relies on verifiable computation (VC) or Zero-knowledge Proof (ZKP) with representative projects as zkSync \cite{zksync} and Aztec \cite{aztec}, which proves the integrity of computations with the proof about validity. And the contract only needs to verify the proof. However, current VC and ZKP use expensive cryptographic primitives \cite{zksnark, bulletproof, spartan, plonk}, which are currently impractical to support the complicated computations like DNN computations. For example, SafetyNets \cite{safetynets} demonstrates to interactively prove the AI computation of only a four-layer CNN. The state-of-the-art work, Slalom \cite{slalom}, has to additionally introduce a trusted execution environment~(TEE) for the DNN computation of MobileNet, ResNet50 and VGG16.

Therefore, the community also develops the second approach, i.e., the proof about fraud, for its potential to support large-scale applications. Representative technologies are Plasma \cite{plasma}, TrueBit\cite{truebit}, Arbitrum \cite{arbitrum}, YODA \cite{yoda}, ACE \cite{ace} and Optimism \cite{optimism}. These technologies introduce a quorum of off-chain nodes with weak assumptions (e.g., AnyTrust), and benefit from the power of pinpoint protocol and the potential reuse of contract checkers \cite{oyente,zeus,securify}. This approach is a promising candidate solution to be a part of Ethereum 2.0 (i.e., the Optimistic Rollup \cite{optimistic}) and will be applied in Uniswap v3 \cite{uniswapv3}. It also inspires other applications such as the efficient 2PC protocol \cite{2pc}.

\para{Referred delegation.}  Some early work investigates the situation where the arbitrator is not a smart contract \cite{quin, versum, focs}. In these work, two or more parties (with at least one honest party) are assumed to execute the same program. Then, the parties can challenge each other with a pinpoint style to locate the disputable step. The step is sent to a judge with weak computation power for arbitration. These technologies do not belong to scalable smart contract, but do share the pinpoint aspect, so are related to our work. %

\section{Conclusions}
We introduce the Agatha system to make decentralized smart contracts capable of DNN computation, which is a goal envisioned by the decentralized computing community. Agatha follows the state-of-the-art verification paradigm that is based on pinpointing and arbitration. The existing technologies rely on VM execution, not supporting native execution, which would make DNN computation too slow to be practical. To enable native execution, we develop Graph-based Pinpoint Protocol (GPP) and Cross-evaluator Consistent Execution (XCE). Agatha is evaluated using popular DNN models, e.g., MobileNet, ResNet50 and VGG16. Compared to the baseline, Agatha's average overhead of off-chain execution is only 3.0\%. By contrast, the slowdown would be at least 620$\times$ for the current VM-based technologies, assuming that these technologies could be built to handle industry-scale DNN models. %

%% file: appendix.tex
\appendix
\section{Understanding the Operations}  \label{sec:understand}
For each model, we look into its operations. As shown in Figure \ref{fig:distribution}, we present the size (i.e., the number of BOPs) of every operation in a small-to-large order, with the size axis in the logarithmic scale. We can see that the operation sizes are not uniformly distributed. Some operations (e.g., \textsf{ConvInteger}) can be several magnitudes larger than the smaller ones (e.g., \textsf{Add}).
\begin{figure}[thb]
  \centering
  \includegraphics[width=0.48\textwidth]{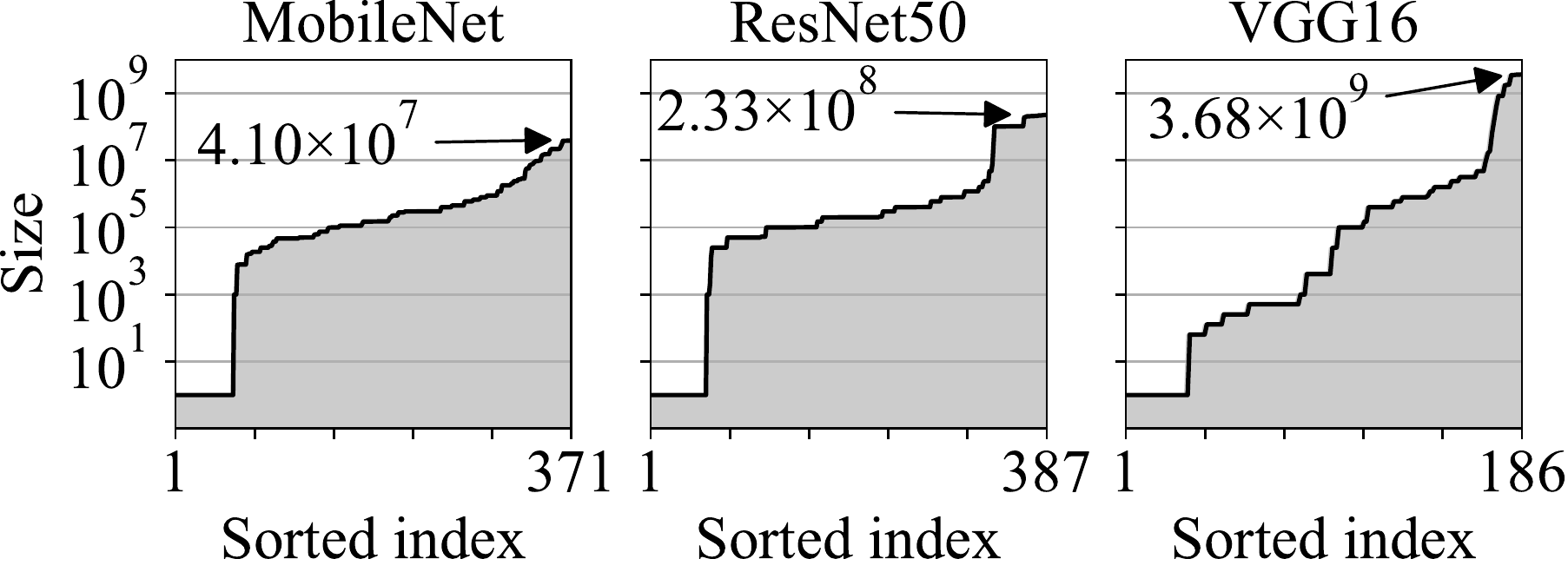}
  \caption{\protect \centering The distribution of operation sizes}\label{fig:distribution}
\end{figure}

In Figure \ref{fig:p0latency}, we illustrate the relation between the operation sizes and the latency of generating operation's MT$^S$. The points in the figure correspond to all the operations in the three models. We can see that, for a large operation, the latency is proportional to its size. For a small operation, the latency is about a small constant, mainly due to the setup time.
\begin{figure}[thb]
  \centering
  \includegraphics[width=0.42\textwidth]{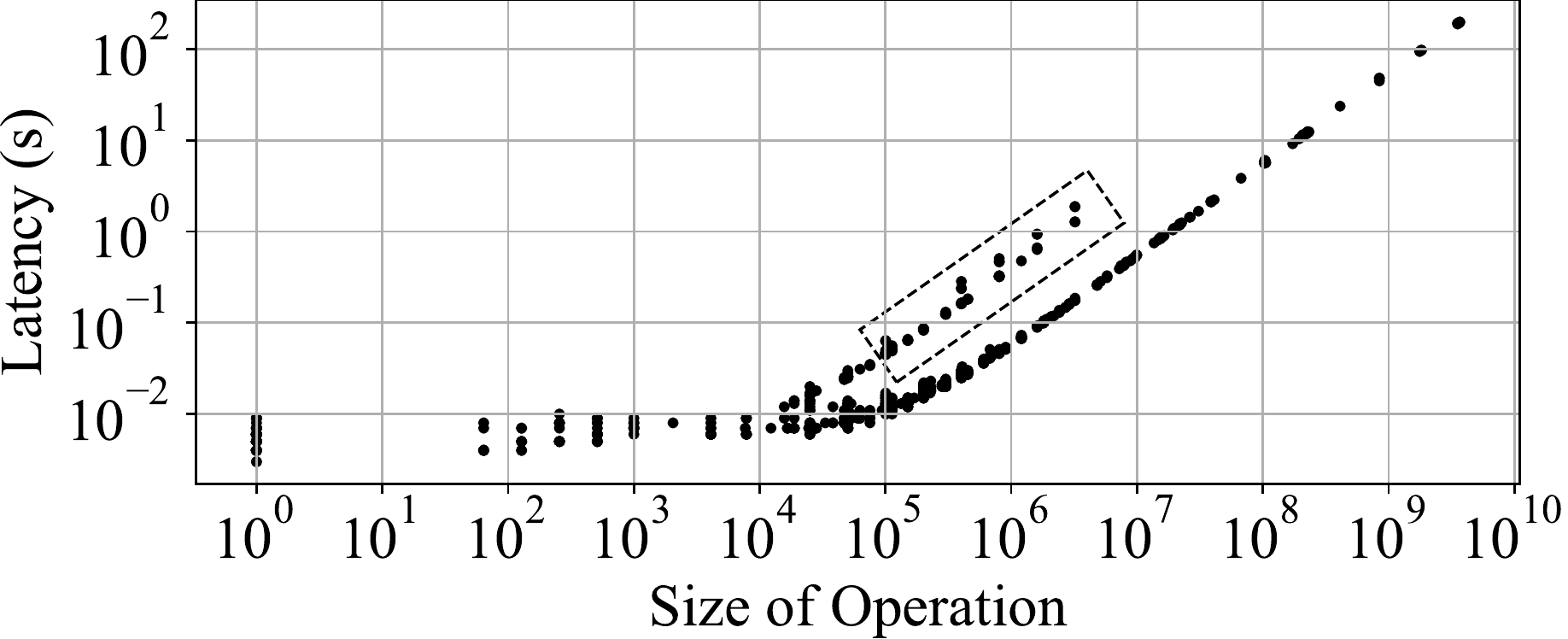}
  \caption{The latency to generate the circuit of operations}
  \label{fig:p0latency}
\end{figure}

Interestingly, we discover that the latency is not only related to the operation sizes, but also the operation types, attributes, and input/output shapes. For example, we highlight the branch which is caused by the longer generation latency of shape broadcasting. Generating all the P2\_MT$^S$ and P1\_MT$^S$ trees for MobileNet, ResNet50, VGG16 take 58s, 444s and 1642s, respectively. Note that this latency is in the protocol's setup period.

%% file: agatha.bbl

\begin{thebibliography}{84}


\ifx \showCODEN    \undefined \def \showCODEN     #1{\unskip}     \fi
\ifx \showDOI      \undefined \def \showDOI       #1{#1}\fi
\ifx \showISBNx    \undefined \def \showISBNx     #1{\unskip}     \fi
\ifx \showISBNxiii \undefined \def \showISBNxiii  #1{\unskip}     \fi
\ifx \showISSN     \undefined \def \showISSN      #1{\unskip}     \fi
\ifx \showLCCN     \undefined \def \showLCCN      #1{\unskip}     \fi
\ifx \shownote     \undefined \def \shownote      #1{#1}          \fi
\ifx \showarticletitle \undefined \def \showarticletitle #1{#1}   \fi
\ifx \showURL      \undefined \def \showURL       {\relax}        \fi
\providecommand\bibfield[2]{#2}
\providecommand\bibinfo[2]{#2}
\providecommand\natexlab[1]{#1}
\providecommand\showeprint[2][]{arXiv:#2}

\bibitem[\protect\citeauthoryear{??}{IEE}{2019}]%
        {IEEE754}
 \bibinfo{year}{2019}\natexlab{}.
\newblock \showarticletitle{IEEE Standard for Floating-Point Arithmetic}.
\newblock \bibinfo{journal}{\emph{IEEE Std 754-2019 (Revision of IEEE
  754-2008)}} (\bibinfo{year}{2019}), \bibinfo{pages}{1--84}.
\newblock
\urldef\tempurl%
\url{https://doi.org/10.1109/IEEESTD.2019.8766229}
\showDOI{\tempurl}


\bibitem[\protect\citeauthoryear{??}{bro}{2021}]%
        {brownie}
 \bibinfo{year}{2021}\natexlab{}.
\newblock \bibinfo{title}{Brownie: a Python-based development and testing
  framework for smart contracts targeting the Ethereum Virtual Machine}.
\newblock
\newblock
\newblock
\shownote{\url{https://github.com/eth-brownie/brownie/}.}


\bibitem[\protect\citeauthoryear{??}{XKC}{2021}]%
        {XKCP}
 \bibinfo{year}{2021}\natexlab{}.
\newblock \bibinfo{title}{The eXtended Keccak Code Package}.
\newblock
\newblock
\newblock
\shownote{\url{https://github.com/XKCP/XKCP}.}


\bibitem[\protect\citeauthoryear{??}{uni}{2021}]%
        {uniswapv3}
 \bibinfo{year}{2021}\natexlab{}.
\newblock \bibinfo{title}{Introducing Uniswap V3}.
\newblock
\newblock
\newblock
\shownote{\url{https://uniswap.org/blog/uniswap-v3/}.}


\bibitem[\protect\citeauthoryear{??}{laz}{2021}]%
        {lazooz}
 \bibinfo{year}{2021}\natexlab{}.
\newblock \bibinfo{title}{La`Zooz}.
\newblock
\newblock
\newblock
\shownote{\url{http://lazooz.org/}.}


\bibitem[\protect\citeauthoryear{??}{onn}{2021a}]%
        {onnxzoo}
 \bibinfo{year}{2021}\natexlab{a}.
\newblock \bibinfo{title}{ONNX Model Zoo}.
\newblock
\newblock
\newblock
\shownote{\url{https://github.com/onnx/models}.}


\bibitem[\protect\citeauthoryear{??}{onn}{2021b}]%
        {onnx}
 \bibinfo{year}{2021}\natexlab{b}.
\newblock \bibinfo{title}{Open Neural Network Exchange}.
\newblock
\newblock
\newblock
\shownote{\url{https://github.com/onnx/onnx}.}


\bibitem[\protect\citeauthoryear{??}{opt}{2021}]%
        {optimistic}
 \bibinfo{year}{2021}\natexlab{}.
\newblock \bibinfo{title}{Optimistic-Rollups for Ethereum}.
\newblock
\newblock
\newblock
\shownote{\url{https://docs.ethhub.io/ethereum-roadmap/layer-2-scaling/optimistic_rollups/}.}


\bibitem[\protect\citeauthoryear{??}{gcc}{2021}]%
        {gcc}
 \bibinfo{year}{2021}\natexlab{}.
\newblock \bibinfo{booktitle}{\emph{Semantics of Floating Point Math in GCC}}.
\newblock
\urldef\tempurl%
\url{https://gcc.gnu.org/wiki/FloatingPointMath}
\showURL{%
\tempurl}


\bibitem[\protect\citeauthoryear{??}{sol}{2021}]%
        {solidity}
 \bibinfo{year}{2021}\natexlab{}.
\newblock \bibinfo{title}{The Solidity Contract-Oriented Programming Language}.
\newblock
\newblock
\newblock
\shownote{\url{https://github.com/ethereum/solidity}.}


\bibitem[\protect\citeauthoryear{Abadi, Barham, Chen, Chen, Davis, Dean, Devin,
  Ghemawat, Irving, Isard, et~al\mbox{.}}{Abadi et~al\mbox{.}}{2016}]%
        {tensorflow}
\bibfield{author}{\bibinfo{person}{Mart{\'\i}n Abadi}, \bibinfo{person}{Paul
  Barham}, \bibinfo{person}{Jianmin Chen}, \bibinfo{person}{Zhifeng Chen},
  \bibinfo{person}{Andy Davis}, \bibinfo{person}{Jeffrey Dean},
  \bibinfo{person}{Matthieu Devin}, \bibinfo{person}{Sanjay Ghemawat},
  \bibinfo{person}{Geoffrey Irving}, \bibinfo{person}{Michael Isard},
  {et~al\mbox{.}}} \bibinfo{year}{2016}\natexlab{}.
\newblock \showarticletitle{Tensorflow: A system for large-scale machine
  learning}. In \bibinfo{booktitle}{\emph{12th $\{$USENIX$\}$ symposium on
  operating systems design and implementation ($\{$OSDI$\}$ 16)}}.
  \bibinfo{pages}{265--283}.
\newblock


\bibitem[\protect\citeauthoryear{ABDK}{ABDK}{2021}]%
        {abdk}
\bibfield{author}{\bibinfo{person}{ABDK}.} \bibinfo{year}{2021}\natexlab{}.
\newblock \bibinfo{title}{ABDK Libraries for Solidity}.
\newblock
\newblock
\newblock
\shownote{\url{https://github.com/abdk-consulting/abdk-libraries-solidity}.}


\bibitem[\protect\citeauthoryear{Ben-Sasson, Chiesa, Tromer, and
  Virza}{Ben-Sasson et~al\mbox{.}}{2014}]%
        {zksnark}
\bibfield{author}{\bibinfo{person}{Eli Ben-Sasson}, \bibinfo{person}{Alessandro
  Chiesa}, \bibinfo{person}{Eran Tromer}, {and} \bibinfo{person}{Madars
  Virza}.} \bibinfo{year}{2014}\natexlab{}.
\newblock \showarticletitle{Succinct non-interactive zero knowledge for a von
  Neumann architecture}. In \bibinfo{booktitle}{\emph{Proceedings of the 23rd
  USENIX conference on Security Symposium}}. \bibinfo{pages}{781--796}.
\newblock


\bibitem[\protect\citeauthoryear{Brassard, Chaum, and Cr{\'e}peau}{Brassard
  et~al\mbox{.}}{1988}]%
        {commitment}
\bibfield{author}{\bibinfo{person}{Gilles Brassard}, \bibinfo{person}{David
  Chaum}, {and} \bibinfo{person}{Claude Cr{\'e}peau}.}
  \bibinfo{year}{1988}\natexlab{}.
\newblock \showarticletitle{Minimum disclosure proofs of knowledge}.
\newblock \bibinfo{journal}{\emph{Journal of computer and system sciences}}
  \bibinfo{volume}{37}, \bibinfo{number}{2} (\bibinfo{year}{1988}),
  \bibinfo{pages}{156--189}.
\newblock


\bibitem[\protect\citeauthoryear{B{\"u}nz, Bootle, Boneh, Poelstra, Wuille, and
  Maxwell}{B{\"u}nz et~al\mbox{.}}{2018}]%
        {bulletproof}
\bibfield{author}{\bibinfo{person}{Benedikt B{\"u}nz},
  \bibinfo{person}{Jonathan Bootle}, \bibinfo{person}{Dan Boneh},
  \bibinfo{person}{Andrew Poelstra}, \bibinfo{person}{Pieter Wuille}, {and}
  \bibinfo{person}{Greg Maxwell}.} \bibinfo{year}{2018}\natexlab{}.
\newblock \showarticletitle{Bulletproofs: Short proofs for confidential
  transactions and more}. In \bibinfo{booktitle}{\emph{2018 IEEE Symposium on
  Security and Privacy (SP)}}. IEEE, \bibinfo{pages}{315--334}.
\newblock


\bibitem[\protect\citeauthoryear{Campanelli, Gennaro, Goldfeder, and
  Nizzardo}{Campanelli et~al\mbox{.}}{2017}]%
        {zkcp}
\bibfield{author}{\bibinfo{person}{Matteo Campanelli}, \bibinfo{person}{Rosario
  Gennaro}, \bibinfo{person}{Steven Goldfeder}, {and} \bibinfo{person}{Luca
  Nizzardo}.} \bibinfo{year}{2017}\natexlab{}.
\newblock \showarticletitle{Zero-knowledge contingent payments revisited:
  Attacks and payments for services}. In \bibinfo{booktitle}{\emph{Proceedings
  of the 2017 ACM SIGSAC Conference on Computer and Communications Security}}.
  \bibinfo{pages}{229--243}.
\newblock


\bibitem[\protect\citeauthoryear{Canetti, Riva, and Rothblum}{Canetti
  et~al\mbox{.}}{2011}]%
        {quin}
\bibfield{author}{\bibinfo{person}{Ran Canetti}, \bibinfo{person}{Ben Riva},
  {and} \bibinfo{person}{Guy~N Rothblum}.} \bibinfo{year}{2011}\natexlab{}.
\newblock \showarticletitle{Practical delegation of computation using multiple
  servers}. In \bibinfo{booktitle}{\emph{Proceedings of the 18th ACM conference
  on Computer and communications security}}. \bibinfo{pages}{445--454}.
\newblock


\bibitem[\protect\citeauthoryear{Chen, Wang, Yan, and Tian}{Chen
  et~al\mbox{.}}{2018}]%
        {cortex}
\bibfield{author}{\bibinfo{person}{Z Chen}, \bibinfo{person}{W Wang},
  \bibinfo{person}{X Yan}, {and} \bibinfo{person}{J Tian}.}
  \bibinfo{year}{2018}\natexlab{}.
\newblock \showarticletitle{Cortex-AI on blockchain: The decentralized AI
  autonomous system}.
\newblock \bibinfo{journal}{\emph{Cortex White Paper}} (\bibinfo{year}{2018}).
\newblock


\bibitem[\protect\citeauthoryear{Das, Ribeiro, and Anand}{Das
  et~al\mbox{.}}{2019}]%
        {yoda}
\bibfield{author}{\bibinfo{person}{Sourav Das}, \bibinfo{person}{Vinay~Joseph
  Ribeiro}, {and} \bibinfo{person}{Abhijeet Anand}.}
  \bibinfo{year}{2019}\natexlab{}.
\newblock \showarticletitle{YODA: Enabling computationally intensive contracts
  on blockchains with Byzantine and Selfish nodes}. In
  \bibinfo{booktitle}{\emph{26th Annual Network and Distributed System Security
  Symposium, {NDSS} 2019}}.
\newblock


\bibitem[\protect\citeauthoryear{Demmel, Ahrens, and Nguyen}{Demmel
  et~al\mbox{.}}{2016}]%
        {reproblas}
\bibfield{author}{\bibinfo{person}{James Demmel}, \bibinfo{person}{Peter
  Ahrens}, {and} \bibinfo{person}{Hong~Diep Nguyen}.}
  \bibinfo{year}{2016}\natexlab{}.
\newblock \bibinfo{booktitle}{\emph{Efficient Reproducible Floating Point
  Summation and BLAS}}.
\newblock \bibinfo{type}{{T}echnical {R}eport} UCB/EECS-2016-121.
  \bibinfo{institution}{EECS Department, University of California, Berkeley}.
\newblock
\urldef\tempurl%
\url{http://www2.eecs.berkeley.edu/Pubs/TechRpts/2016/EECS-2016-121.html}
\showURL{%
\tempurl}


\bibitem[\protect\citeauthoryear{Deng, Dong, Socher, Li, Li, and Fei-Fei}{Deng
  et~al\mbox{.}}{2009}]%
        {imagenet}
\bibfield{author}{\bibinfo{person}{Jia Deng}, \bibinfo{person}{Wei Dong},
  \bibinfo{person}{Richard Socher}, \bibinfo{person}{Li-Jia Li},
  \bibinfo{person}{Kai Li}, {and} \bibinfo{person}{Li Fei-Fei}.}
  \bibinfo{year}{2009}\natexlab{}.
\newblock \showarticletitle{Imagenet: A large-scale hierarchical image
  database}. In \bibinfo{booktitle}{\emph{2009 IEEE conference on computer
  vision and pattern recognition}}. Ieee, \bibinfo{pages}{248--255}.
\newblock


\bibitem[\protect\citeauthoryear{Dziembowski, Eckey, and Faust}{Dziembowski
  et~al\mbox{.}}{2018a}]%
        {fairswap}
\bibfield{author}{\bibinfo{person}{Stefan Dziembowski}, \bibinfo{person}{Lisa
  Eckey}, {and} \bibinfo{person}{Sebastian Faust}.}
  \bibinfo{year}{2018}\natexlab{a}.
\newblock \showarticletitle{Fairswap: How to fairly exchange digital goods}. In
  \bibinfo{booktitle}{\emph{Proceedings of the 2018 ACM SIGSAC Conference on
  Computer and Communications Security}}. ACM, \bibinfo{pages}{967--984}.
\newblock


\bibitem[\protect\citeauthoryear{Dziembowski, Faust, and
  Host{\'a}kov{\'a}}{Dziembowski et~al\mbox{.}}{2018b}]%
        {statechannel}
\bibfield{author}{\bibinfo{person}{Stefan Dziembowski},
  \bibinfo{person}{Sebastian Faust}, {and} \bibinfo{person}{Kristina
  Host{\'a}kov{\'a}}.} \bibinfo{year}{2018}\natexlab{b}.
\newblock \showarticletitle{General state channel networks}. In
  \bibinfo{booktitle}{\emph{Proceedings of the 2018 ACM SIGSAC Conference on
  Computer and Communications Security}}. \bibinfo{pages}{949--966}.
\newblock


\bibitem[\protect\citeauthoryear{Eberhardt and Tai}{Eberhardt and Tai}{2018}]%
        {zokrates}
\bibfield{author}{\bibinfo{person}{Jacob Eberhardt} {and}
  \bibinfo{person}{Stefan Tai}.} \bibinfo{year}{2018}\natexlab{}.
\newblock \showarticletitle{ZoKrates-Scalable Privacy-Preserving Off-Chain
  Computations}. In \bibinfo{booktitle}{\emph{IEEE International Conference on
  Blockchain. IEEE}}.
\newblock


\bibitem[\protect\citeauthoryear{Eckey, Faust, and Schlosser}{Eckey
  et~al\mbox{.}}{2020}]%
        {optiswap}
\bibfield{author}{\bibinfo{person}{Lisa Eckey}, \bibinfo{person}{Sebastian
  Faust}, {and} \bibinfo{person}{Benjamin Schlosser}.}
  \bibinfo{year}{2020}\natexlab{}.
\newblock \showarticletitle{Optiswap: Fast optimistic fair exchange}. In
  \bibinfo{booktitle}{\emph{Proceedings of the 15th ACM Asia Conference on
  Computer and Communications Security}}. \bibinfo{pages}{543--557}.
\newblock


\bibitem[\protect\citeauthoryear{Feige and Kilian}{Feige and Kilian}{1997}]%
        {focs}
\bibfield{author}{\bibinfo{person}{Uriel Feige} {and} \bibinfo{person}{Joe
  Kilian}.} \bibinfo{year}{1997}\natexlab{}.
\newblock \showarticletitle{Making games short}. In
  \bibinfo{booktitle}{\emph{Proceedings of the twenty-ninth annual ACM
  symposium on Theory of computing}}. \bibinfo{pages}{506--516}.
\newblock


\bibitem[\protect\citeauthoryear{Foundation}{Foundation}{2021a}]%
        {ethereum}
\bibfield{author}{\bibinfo{person}{Ethereum Foundation}.}
  \bibinfo{year}{2021}\natexlab{a}.
\newblock \bibinfo{title}{Ethereum}.
\newblock
\newblock
\newblock
\shownote{\url{https://www.ethereum.org/}.}


\bibitem[\protect\citeauthoryear{Foundation}{Foundation}{2021b}]%
        {istanbul}
\bibfield{author}{\bibinfo{person}{Ethereum Foundation}.}
  \bibinfo{year}{2021}\natexlab{b}.
\newblock \bibinfo{title}{The Ethereum Istanbul hard fork}.
\newblock
\newblock
\newblock
\shownote{\url{https://eth.wiki/en/roadmap/istanbul}.}


\bibitem[\protect\citeauthoryear{Foundation}{Foundation}{2021c}]%
        {evm}
\bibfield{author}{\bibinfo{person}{Ethereum Foundation}.}
  \bibinfo{year}{2021}\natexlab{c}.
\newblock \bibinfo{title}{Ethereum Virtual Machine (EVM)}.
\newblock
\newblock
\newblock
\shownote{\url{https://ethereum.org/en/developers/docs/evm/}.}


\bibitem[\protect\citeauthoryear{Gabizon, Williamson, and Ciobotaru}{Gabizon
  et~al\mbox{.}}{2019}]%
        {plonk}
\bibfield{author}{\bibinfo{person}{Ariel Gabizon}, \bibinfo{person}{Zachary~J
  Williamson}, {and} \bibinfo{person}{Oana Ciobotaru}.}
  \bibinfo{year}{2019}\natexlab{}.
\newblock \showarticletitle{PLONK: Permutations over Lagrange-bases for
  Oecumenical Noninteractive arguments of Knowledge.}
\newblock \bibinfo{journal}{\emph{IACR Cryptol. ePrint Arch.}}
  \bibinfo{volume}{2019} (\bibinfo{year}{2019}), \bibinfo{pages}{953}.
\newblock


\bibitem[\protect\citeauthoryear{Ghodsi, Gu, and Garg}{Ghodsi
  et~al\mbox{.}}{2017}]%
        {safetynets}
\bibfield{author}{\bibinfo{person}{Zahra Ghodsi}, \bibinfo{person}{Tianyu Gu},
  {and} \bibinfo{person}{Siddharth Garg}.} \bibinfo{year}{2017}\natexlab{}.
\newblock \showarticletitle{Safetynets: Verifiable execution of deep neural
  networks on an untrusted cloud}.
\newblock \bibinfo{journal}{\emph{Advances in Neural Information Processing
  Systems}}  \bibinfo{volume}{30} (\bibinfo{year}{2017}),
  \bibinfo{pages}{4672--4681}.
\newblock


\bibitem[\protect\citeauthoryear{Goertzel, Giacomelli, Hanson, Pennachin, and
  Argentieri}{Goertzel et~al\mbox{.}}{2017}]%
        {singularitynet}
\bibfield{author}{\bibinfo{person}{Ben Goertzel}, \bibinfo{person}{Simone
  Giacomelli}, \bibinfo{person}{David Hanson}, \bibinfo{person}{Cassio
  Pennachin}, {and} \bibinfo{person}{Marco Argentieri}.}
  \bibinfo{year}{2017}\natexlab{}.
\newblock \showarticletitle{SingularityNET: A decentralized, open market and
  inter-network for AIs}.
\newblock \bibinfo{journal}{\emph{Thoughts, Theories \& Studies on Artificial
  Intelligence (AI). Research}} (\bibinfo{year}{2017}).
\newblock


\bibitem[\protect\citeauthoryear{Goldberg}{Goldberg}{1991}]%
        {goldberg1991every}
\bibfield{author}{\bibinfo{person}{David Goldberg}.}
  \bibinfo{year}{1991}\natexlab{}.
\newblock \showarticletitle{What every computer scientist should know about
  floating-point arithmetic}.
\newblock \bibinfo{journal}{\emph{ACM Computing Surveys (CSUR)}}
  \bibinfo{volume}{23}, \bibinfo{number}{1} (\bibinfo{year}{1991}),
  \bibinfo{pages}{5--48}.
\newblock


\bibitem[\protect\citeauthoryear{Group}{Group}{2021}]%
        {wasm}
\bibfield{author}{\bibinfo{person}{W3C~Community Group}.}
  \bibinfo{year}{2021}\natexlab{}.
\newblock \bibinfo{title}{WebAssembly}.
\newblock
\newblock
\newblock
\shownote{\url{https://webassembly.org/}.}


\bibitem[\protect\citeauthoryear{Gulley, Gopal, Yap, Feghali, Guilford, and
  Wolrich}{Gulley et~al\mbox{.}}{2013}]%
        {sha}
\bibfield{author}{\bibinfo{person}{Sean Gulley}, \bibinfo{person}{Vinodh
  Gopal}, \bibinfo{person}{Kirk Yap}, \bibinfo{person}{Wajdi Feghali},
  \bibinfo{person}{Jim Guilford}, {and} \bibinfo{person}{Gil Wolrich}.}
  \bibinfo{year}{2013}\natexlab{}.
\newblock \showarticletitle{Intel sha extensions--new instructions supporting
  the secure hash algorithm on intel architecture processor}.
\newblock \bibinfo{journal}{\emph{Intel White Paper}} (\bibinfo{year}{2013}).
\newblock


\bibitem[\protect\citeauthoryear{Harris and Waggoner}{Harris and
  Waggoner}{2019}]%
        {microsoft}
\bibfield{author}{\bibinfo{person}{Justin~D Harris} {and} \bibinfo{person}{Bo
  Waggoner}.} \bibinfo{year}{2019}\natexlab{}.
\newblock \showarticletitle{Decentralized and collaborative ai on blockchain}.
  In \bibinfo{booktitle}{\emph{2019 IEEE International Conference on Blockchain
  (Blockchain)}}. IEEE, \bibinfo{pages}{368--375}.
\newblock


\bibitem[\protect\citeauthoryear{Harz, Gudgeon, Gervais, and Knottenbelt}{Harz
  et~al\mbox{.}}{2019}]%
        {balance}
\bibfield{author}{\bibinfo{person}{Dominik Harz}, \bibinfo{person}{Lewis
  Gudgeon}, \bibinfo{person}{Arthur Gervais}, {and} \bibinfo{person}{William~J
  Knottenbelt}.} \bibinfo{year}{2019}\natexlab{}.
\newblock \showarticletitle{Balance: Dynamic adjustment of cryptocurrency
  deposits}. In \bibinfo{booktitle}{\emph{Proceedings of the 2019 ACM SIGSAC
  Conference on Computer and Communications Security}}.
  \bibinfo{pages}{1485--1502}.
\newblock


\bibitem[\protect\citeauthoryear{He, Zhang, Ren, and Sun}{He
  et~al\mbox{.}}{2016}]%
        {resnet}
\bibfield{author}{\bibinfo{person}{Kaiming He}, \bibinfo{person}{Xiangyu
  Zhang}, \bibinfo{person}{Shaoqing Ren}, {and} \bibinfo{person}{Jian Sun}.}
  \bibinfo{year}{2016}\natexlab{}.
\newblock \showarticletitle{Identity mappings in deep residual networks}. In
  \bibinfo{booktitle}{\emph{European conference on computer vision}}. Springer,
  \bibinfo{pages}{630--645}.
\newblock


\bibitem[\protect\citeauthoryear{Herman}{Herman}{2021}]%
        {oraclize}
\bibfield{author}{\bibinfo{person}{Joshua Herman}.}
  \bibinfo{year}{2021}\natexlab{}.
\newblock \bibinfo{title}{Make your Sklearn models accessible to Smart
  Contracts using Oracalize}.
\newblock
\newblock
\newblock
\shownote{\url{https://medium.com/zitterbewegung/4d9593c80383}.}


\bibitem[\protect\citeauthoryear{Inc.}{Inc.}{2021}]%
        {ganache}
\bibfield{author}{\bibinfo{person}{ConsenSys~Software Inc.}}
  \bibinfo{year}{2021}\natexlab{}.
\newblock \bibinfo{title}{One Click Blockchain}.
\newblock
\newblock
\newblock
\shownote{\url{https://www.trufflesuite.com/ganache}.}


\bibitem[\protect\citeauthoryear{Inc.}{Inc.}{2018}]%
        {cryptokitties}
\bibfield{author}{\bibinfo{person}{Dapper~Labs Inc.}}
  \bibinfo{year}{2018}\natexlab{}.
\newblock \bibinfo{title}{CryptoKitties: Collect and breed digital cats!}
\newblock
\newblock
\newblock
\shownote{\url{https://www.cryptokitties.co/}.}


\bibitem[\protect\citeauthoryear{Jacob, Kligys, Chen, Zhu, Tang, Howard, Adam,
  and Kalenichenko}{Jacob et~al\mbox{.}}{2018}]%
        {jacob2018quantization}
\bibfield{author}{\bibinfo{person}{Benoit Jacob}, \bibinfo{person}{Skirmantas
  Kligys}, \bibinfo{person}{Bo Chen}, \bibinfo{person}{Menglong Zhu},
  \bibinfo{person}{Matthew Tang}, \bibinfo{person}{Andrew Howard},
  \bibinfo{person}{Hartwig Adam}, {and} \bibinfo{person}{Dmitry Kalenichenko}.}
  \bibinfo{year}{2018}\natexlab{}.
\newblock \showarticletitle{Quantization and training of neural networks for
  efficient integer-arithmetic-only inference}. In
  \bibinfo{booktitle}{\emph{Proceedings of the IEEE Conference on Computer
  Vision and Pattern Recognition}}. \bibinfo{pages}{2704--2713}.
\newblock


\bibitem[\protect\citeauthoryear{Kahn}{Kahn}{1962}]%
        {ts}
\bibfield{author}{\bibinfo{person}{Arthur~B Kahn}.}
  \bibinfo{year}{1962}\natexlab{}.
\newblock \showarticletitle{Topological sorting of large networks}.
\newblock \bibinfo{journal}{\emph{Commun. ACM}} \bibinfo{volume}{5},
  \bibinfo{number}{11} (\bibinfo{year}{1962}), \bibinfo{pages}{558--562}.
\newblock


\bibitem[\protect\citeauthoryear{Kalodner, Goldfeder, Chen, Weinberg, and
  Felten}{Kalodner et~al\mbox{.}}{2018}]%
        {arbitrum}
\bibfield{author}{\bibinfo{person}{Harry Kalodner}, \bibinfo{person}{Steven
  Goldfeder}, \bibinfo{person}{Xiaoqi Chen}, \bibinfo{person}{S~Matthew
  Weinberg}, {and} \bibinfo{person}{Edward~W Felten}.}
  \bibinfo{year}{2018}\natexlab{}.
\newblock \showarticletitle{Arbitrum: Scalable, private smart contracts}. In
  \bibinfo{booktitle}{\emph{Proceedings of the 27th USENIX Conference on
  Security Symposium}}. USENIX Association, \bibinfo{pages}{1353--1370}.
\newblock


\bibitem[\protect\citeauthoryear{Kalra, Goel, Dhawan, and Sharma}{Kalra
  et~al\mbox{.}}{2018}]%
        {zeus}
\bibfield{author}{\bibinfo{person}{Sukrit Kalra}, \bibinfo{person}{Seep Goel},
  \bibinfo{person}{Mohan Dhawan}, {and} \bibinfo{person}{Subodh Sharma}.}
  \bibinfo{year}{2018}\natexlab{}.
\newblock \showarticletitle{ZEUS: Analyzing Safety of Smart Contracts}. In
  \bibinfo{booktitle}{\emph{25th Annual Network and Distributed System Security
  Symposium, {NDSS} 2018, San Diego, California, USA, February 18-21, 2018}}.
\newblock


\bibitem[\protect\citeauthoryear{Kurtulmus and Daniel}{Kurtulmus and
  Daniel}{2018}]%
        {algorithmia}
\bibfield{author}{\bibinfo{person}{A~Besir Kurtulmus} {and}
  \bibinfo{person}{Kenny Daniel}.} \bibinfo{year}{2018}\natexlab{}.
\newblock \showarticletitle{Trustless machine learning contracts; evaluating
  and exchanging machine learning models on the ethereum blockchain}.
\newblock \bibinfo{journal}{\emph{arXiv preprint arXiv:1802.10185}}
  (\bibinfo{year}{2018}).
\newblock


\bibitem[\protect\citeauthoryear{Labs}{Labs}{2021a}]%
        {zksync}
\bibfield{author}{\bibinfo{person}{Matter Labs}.}
  \bibinfo{year}{2021}\natexlab{a}.
\newblock \bibinfo{title}{zkSync}.
\newblock
\newblock
\newblock
\shownote{\url{https://zksync.io/}.}


\bibitem[\protect\citeauthoryear{Labs}{Labs}{2021b}]%
        {ipfs}
\bibfield{author}{\bibinfo{person}{Protocol Labs}.}
  \bibinfo{year}{2021}\natexlab{b}.
\newblock \bibinfo{title}{IPFS}.
\newblock
\newblock
\newblock
\shownote{\url{https://ipfs.io/}.}


\bibitem[\protect\citeauthoryear{Labs}{Labs}{2021c}]%
        {uniswap}
\bibfield{author}{\bibinfo{person}{Uniswap Labs}.}
  \bibinfo{year}{2021}\natexlab{c}.
\newblock \bibinfo{title}{Uniswap:Decentralized Trading Protocol}.
\newblock
\newblock
\newblock
\shownote{\url{https://uniswap.org/}.}


\bibitem[\protect\citeauthoryear{Laptev, Yosinski, Li, and Smyl}{Laptev
  et~al\mbox{.}}{2017}]%
        {uber}
\bibfield{author}{\bibinfo{person}{Nikolay Laptev}, \bibinfo{person}{Jason
  Yosinski}, \bibinfo{person}{Li~Erran Li}, {and} \bibinfo{person}{Slawek
  Smyl}.} \bibinfo{year}{2017}\natexlab{}.
\newblock \showarticletitle{Time-series extreme event forecasting with neural
  networks at uber}. In \bibinfo{booktitle}{\emph{International conference on
  machine learning}}, Vol.~\bibinfo{volume}{34}. \bibinfo{pages}{1--5}.
\newblock


\bibitem[\protect\citeauthoryear{Ltd.}{Ltd.}{2020}]%
        {aztec}
\bibfield{author}{\bibinfo{person}{Spilsbury~Holdings Ltd.}}
  \bibinfo{year}{2020}\natexlab{}.
\newblock \bibinfo{title}{Aztec}.
\newblock
\newblock
\newblock
\shownote{\url{https://aztec.network/}.}


\bibitem[\protect\citeauthoryear{Luu, Chu, Olickel, Saxena, and Hobor}{Luu
  et~al\mbox{.}}{2016}]%
        {oyente}
\bibfield{author}{\bibinfo{person}{Loi Luu}, \bibinfo{person}{Duc-Hiep Chu},
  \bibinfo{person}{Hrishi Olickel}, \bibinfo{person}{Prateek Saxena}, {and}
  \bibinfo{person}{Aquinas Hobor}.} \bibinfo{year}{2016}\natexlab{}.
\newblock \showarticletitle{Making smart contracts smarter}. In
  \bibinfo{booktitle}{\emph{Proceedings of the 2016 ACM SIGSAC Conference on
  Computer and Communications Security}}. ACM, \bibinfo{pages}{254--269}.
\newblock


\bibitem[\protect\citeauthoryear{Malavolta, Moreno-Sanchez, Kate, Maffei, and
  Ravi}{Malavolta et~al\mbox{.}}{2017}]%
        {payment}
\bibfield{author}{\bibinfo{person}{Giulio Malavolta}, \bibinfo{person}{Pedro
  Moreno-Sanchez}, \bibinfo{person}{Aniket Kate}, \bibinfo{person}{Matteo
  Maffei}, {and} \bibinfo{person}{Srivatsan Ravi}.}
  \bibinfo{year}{2017}\natexlab{}.
\newblock \showarticletitle{Concurrency and privacy with payment-channel
  networks}. In \bibinfo{booktitle}{\emph{Proceedings of the 2017 ACM SIGSAC
  Conference on Computer and Communications Security}}.
  \bibinfo{pages}{455--471}.
\newblock


\bibitem[\protect\citeauthoryear{Marcus, Heilman, and Goldberg}{Marcus
  et~al\mbox{.}}{2018}]%
        {eclipse}
\bibfield{author}{\bibinfo{person}{Yuval Marcus}, \bibinfo{person}{Ethan
  Heilman}, {and} \bibinfo{person}{Sharon Goldberg}.}
  \bibinfo{year}{2018}\natexlab{}.
\newblock \showarticletitle{Low-Resource Eclipse Attacks on Ethereum's
  Peer-to-Peer Network.}
\newblock \bibinfo{journal}{\emph{IACR Cryptol. ePrint Arch.}}
  \bibinfo{volume}{2018} (\bibinfo{year}{2018}), \bibinfo{pages}{236}.
\newblock


\bibitem[\protect\citeauthoryear{McCorry, Shahandashti, and Hao}{McCorry
  et~al\mbox{.}}{2017}]%
        {vote}
\bibfield{author}{\bibinfo{person}{Patrick McCorry}, \bibinfo{person}{Siamak~F
  Shahandashti}, {and} \bibinfo{person}{Feng Hao}.}
  \bibinfo{year}{2017}\natexlab{}.
\newblock \showarticletitle{A smart contract for boardroom voting with maximum
  voter privacy}. In \bibinfo{booktitle}{\emph{International Conference on
  Financial Cryptography and Data Security}}. Springer,
  \bibinfo{pages}{357--375}.
\newblock


\bibitem[\protect\citeauthoryear{Mendis, Sabounchi, Wei, and Roche}{Mendis
  et~al\mbox{.}}{2018}]%
        {cdda}
\bibfield{author}{\bibinfo{person}{Gihan~J Mendis}, \bibinfo{person}{Moein
  Sabounchi}, \bibinfo{person}{Jin Wei}, {and} \bibinfo{person}{Rigoberto
  Roche}.} \bibinfo{year}{2018}\natexlab{}.
\newblock \showarticletitle{Blockchain as a service: an autonomous, privacy
  preserving, decentralized architecture for deep learning}.
\newblock \bibinfo{journal}{\emph{arXiv preprint arXiv:1807.02515}}
  (\bibinfo{year}{2018}).
\newblock


\bibitem[\protect\citeauthoryear{Merkle}{Merkle}{1987}]%
        {mt}
\bibfield{author}{\bibinfo{person}{Ralph~C Merkle}.}
  \bibinfo{year}{1987}\natexlab{}.
\newblock \showarticletitle{A digital signature based on a conventional
  encryption function}. In \bibinfo{booktitle}{\emph{Conference on the theory
  and application of cryptographic techniques}}. Springer,
  \bibinfo{pages}{369--378}.
\newblock


\bibitem[\protect\citeauthoryear{Microsoft}{Microsoft}{2021}]%
        {onnxruntime}
\bibfield{author}{\bibinfo{person}{Microsoft}.}
  \bibinfo{year}{2021}\natexlab{}.
\newblock \bibinfo{title}{ONNX Runtime}.
\newblock
\newblock
\newblock
\shownote{\url{https://github.com/microsoft/onnxruntime}.}


\bibitem[\protect\citeauthoryear{Monniaux}{Monniaux}{2008}]%
        {monniaux2008pitfalls}
\bibfield{author}{\bibinfo{person}{David Monniaux}.}
  \bibinfo{year}{2008}\natexlab{}.
\newblock \showarticletitle{The pitfalls of verifying floating-point
  computations}.
\newblock \bibinfo{journal}{\emph{ACM Transactions on Programming Languages and
  Systems (TOPLAS)}} \bibinfo{volume}{30}, \bibinfo{number}{3}
  (\bibinfo{year}{2008}), \bibinfo{pages}{1--41}.
\newblock


\bibitem[\protect\citeauthoryear{of~State}{of~State}{2021}]%
        {ballot}
\bibfield{author}{\bibinfo{person}{California~Secretary of State}.}
  \bibinfo{year}{2021}\natexlab{}.
\newblock \bibinfo{title}{Signature Verification, Ballot Processing, and Ballot
  Counting (Emergency Regulations)}.
\newblock
\newblock
\newblock
\shownote{\url{https://www.sos.ca.gov/administration/regulations/current-regulations/elections/signature-verification-ballot-processing-and-ballot-counting-emergency-regulations}.}


\bibitem[\protect\citeauthoryear{Paszke, Gross, Massa, Lerer, Bradbury, Chanan,
  Killeen, Lin, Gimelshein, Antiga, et~al\mbox{.}}{Paszke
  et~al\mbox{.}}{2019}]%
        {pytorch}
\bibfield{author}{\bibinfo{person}{Adam Paszke}, \bibinfo{person}{Sam Gross},
  \bibinfo{person}{Francisco Massa}, \bibinfo{person}{Adam Lerer},
  \bibinfo{person}{James Bradbury}, \bibinfo{person}{Gregory Chanan},
  \bibinfo{person}{Trevor Killeen}, \bibinfo{person}{Zeming Lin},
  \bibinfo{person}{Natalia Gimelshein}, \bibinfo{person}{Luca Antiga},
  {et~al\mbox{.}}} \bibinfo{year}{2019}\natexlab{}.
\newblock \showarticletitle{Pytorch: An imperative style, high-performance deep
  learning library}. In \bibinfo{booktitle}{\emph{Advances in neural
  information processing systems}}. \bibinfo{pages}{8026--8037}.
\newblock


\bibitem[\protect\citeauthoryear{PBC}{PBC}{2021}]%
        {optimism}
\bibfield{author}{\bibinfo{person}{Optimism PBC}.}
  \bibinfo{year}{2021}\natexlab{}.
\newblock \bibinfo{title}{Optimism}.
\newblock
\newblock
\newblock
\shownote{\url{https://optimism.io/}.}


\bibitem[\protect\citeauthoryear{Perez and Livshits}{Perez and
  Livshits}{2020}]%
        {perez2019broken}
\bibfield{author}{\bibinfo{person}{Daniel Perez} {and}
  \bibinfo{person}{Benjamin Livshits}.} \bibinfo{year}{2020}\natexlab{}.
\newblock \showarticletitle{Broken Metre: Attacking Resource Metering in
  {EVM}}. In \bibinfo{booktitle}{\emph{27th Annual Network and Distributed
  System Security Symposium, {NDSS} 2020}}. \bibinfo{publisher}{The Internet
  Society}.
\newblock


\bibitem[\protect\citeauthoryear{Poon and Buterin}{Poon and Buterin}{2017}]%
        {plasma}
\bibfield{author}{\bibinfo{person}{Joseph Poon} {and} \bibinfo{person}{Vitalik
  Buterin}.} \bibinfo{year}{2017}\natexlab{}.
\newblock \showarticletitle{Plasma: Scalable autonomous smart contracts}.
\newblock \bibinfo{journal}{\emph{White paper}} (\bibinfo{year}{2017}),
  \bibinfo{pages}{1--47}.
\newblock


\bibitem[\protect\citeauthoryear{Rodriguez}{Rodriguez}{2021}]%
        {iamm}
\bibfield{author}{\bibinfo{person}{Jesus Rodriguez}.}
  \bibinfo{year}{2021}\natexlab{}.
\newblock \bibinfo{title}{When DeFi Becomes Intelligent}.
\newblock
\newblock
\newblock
\shownote{\url{https://www.coindesk.com/when-defi-becomes-intelligent}.}


\bibitem[\protect\citeauthoryear{Sandler, Howard, Zhu, Zhmoginov, and
  Chen}{Sandler et~al\mbox{.}}{2018}]%
        {mobilenet}
\bibfield{author}{\bibinfo{person}{Mark Sandler}, \bibinfo{person}{Andrew
  Howard}, \bibinfo{person}{Menglong Zhu}, \bibinfo{person}{Andrey Zhmoginov},
  {and} \bibinfo{person}{Liang-Chieh Chen}.} \bibinfo{year}{2018}\natexlab{}.
\newblock \showarticletitle{Mobilenetv2: Inverted residuals and linear
  bottlenecks}. In \bibinfo{booktitle}{\emph{Proceedings of the IEEE conference
  on computer vision and pattern recognition}}. \bibinfo{pages}{4510--4520}.
\newblock


\bibitem[\protect\citeauthoryear{Seifelnasr, Galal, and Youssef}{Seifelnasr
  et~al\mbox{.}}{2020}]%
        {vote2}
\bibfield{author}{\bibinfo{person}{Mohamed Seifelnasr},
  \bibinfo{person}{Hisham~S Galal}, {and} \bibinfo{person}{Amr~M Youssef}.}
  \bibinfo{year}{2020}\natexlab{}.
\newblock \showarticletitle{Scalable open-vote network on ethereum}. In
  \bibinfo{booktitle}{\emph{International Conference on Financial Cryptography
  and Data Security}}. Springer, \bibinfo{pages}{436--450}.
\newblock


\bibitem[\protect\citeauthoryear{Setty}{Setty}{2020}]%
        {spartan}
\bibfield{author}{\bibinfo{person}{Srinath Setty}.}
  \bibinfo{year}{2020}\natexlab{}.
\newblock \showarticletitle{Spartan: Efficient and general-purpose zkSNARKs
  without trusted setup}. In \bibinfo{booktitle}{\emph{Annual International
  Cryptology Conference}}. Springer, \bibinfo{pages}{704--737}.
\newblock


\bibitem[\protect\citeauthoryear{Simonyan and Zisserman}{Simonyan and
  Zisserman}{2015}]%
        {vgg}
\bibfield{author}{\bibinfo{person}{Karen Simonyan} {and}
  \bibinfo{person}{Andrew Zisserman}.} \bibinfo{year}{2015}\natexlab{}.
\newblock \showarticletitle{Very Deep Convolutional Networks for Large-Scale
  Image Recognition}. In \bibinfo{booktitle}{\emph{3rd International Conference
  on Learning Representations, {ICLR} 2015}}.
\newblock


\bibitem[\protect\citeauthoryear{Singla, Bose, and Katariya}{Singla
  et~al\mbox{.}}{2018}]%
        {smarthome}
\bibfield{author}{\bibinfo{person}{Kushal Singla}, \bibinfo{person}{Joy Bose},
  {and} \bibinfo{person}{Sharvil Katariya}.} \bibinfo{year}{2018}\natexlab{}.
\newblock \showarticletitle{Machine learning for secure device personalization
  using blockchain}. In \bibinfo{booktitle}{\emph{2018 International Conference
  on Advances in Computing, Communications and Informatics (ICACCI)}}. IEEE,
  \bibinfo{pages}{67--73}.
\newblock


\bibitem[\protect\citeauthoryear{Spooner, Fearnley, Savani, and
  Koukorinis}{Spooner et~al\mbox{.}}{2018}]%
        {amm2}
\bibfield{author}{\bibinfo{person}{Thomas Spooner}, \bibinfo{person}{John
  Fearnley}, \bibinfo{person}{Rahul Savani}, {and} \bibinfo{person}{Andreas
  Koukorinis}.} \bibinfo{year}{2018}\natexlab{}.
\newblock \showarticletitle{Market Making via Reinforcement Learning}. In
  \bibinfo{booktitle}{\emph{Proceedings of the 17th International Conference on
  Autonomous Agents and MultiAgent Systems}}. \bibinfo{pages}{434--442}.
\newblock


\bibitem[\protect\citeauthoryear{Spooner and Savani}{Spooner and
  Savani}{2020}]%
        {amm1}
\bibfield{author}{\bibinfo{person}{Thomas Spooner} {and} \bibinfo{person}{Rahul
  Savani}.} \bibinfo{year}{2020}\natexlab{}.
\newblock \showarticletitle{Robust Market Making via Adversarial Reinforcement
  Learning}. In \bibinfo{booktitle}{\emph{Proceedings of the Twenty-Ninth
  International Joint Conference on Artificial Intelligence, {IJCAI-20}}},
  \bibfield{editor}{\bibinfo{person}{Christian Bessiere}} (Ed.).
  \bibinfo{publisher}{International Joint Conferences on Artificial
  Intelligence Organization}, \bibinfo{pages}{4590--4596}.
\newblock
\urldef\tempurl%
\url{https://doi.org/10.24963/ijcai.2020/633}
\showDOI{\tempurl}
\newblock
\shownote{Special Track on AI in FinTech.}


\bibitem[\protect\citeauthoryear{Sze, Chen, Yang, and Emer}{Sze
  et~al\mbox{.}}{2017}]%
        {sze2017efficient}
\bibfield{author}{\bibinfo{person}{Vivienne Sze}, \bibinfo{person}{Yu-Hsin
  Chen}, \bibinfo{person}{Tien-Ju Yang}, {and} \bibinfo{person}{Joel~S Emer}.}
  \bibinfo{year}{2017}\natexlab{}.
\newblock \showarticletitle{Efficient processing of deep neural networks: A
  tutorial and survey}.
\newblock \bibinfo{journal}{\emph{Proc. IEEE}} \bibinfo{volume}{105},
  \bibinfo{number}{12} (\bibinfo{year}{2017}), \bibinfo{pages}{2295--2329}.
\newblock


\bibitem[\protect\citeauthoryear{Tang}{Tang}{2019}]%
        {EIP2200}
\bibfield{author}{\bibinfo{person}{Wei Tang}.} \bibinfo{year}{July
  2019}\natexlab{}.
\newblock \showarticletitle{EIP-2200: Structured Definitions for Net Gas
  Metering}.
\newblock \bibinfo{journal}{\emph{Ethereum Improvement Proposals}}
  \bibinfo{volume}{no. 2200} (\bibinfo{year}{July 2019}).
\newblock
\newblock
\shownote{[Online serial]. Available:
  \url{https://eips.ethereum.org/EIPS/eip-2200}.}


\bibitem[\protect\citeauthoryear{Tapscott and Tapscott}{Tapscott and
  Tapscott}{2016}]%
        {book}
\bibfield{author}{\bibinfo{person}{Don Tapscott} {and} \bibinfo{person}{Alex
  Tapscott}.} \bibinfo{year}{2016}\natexlab{}.
\newblock \bibinfo{booktitle}{\emph{Blockchain revolution: how the technology
  behind bitcoin is changing money, business, and the world}}.
\newblock \bibinfo{publisher}{Penguin}.
\newblock


\bibitem[\protect\citeauthoryear{Teutsch and Reitwie{\ss}ner}{Teutsch and
  Reitwie{\ss}ner}{2019}]%
        {truebit}
\bibfield{author}{\bibinfo{person}{Jason Teutsch} {and}
  \bibinfo{person}{Christian Reitwie{\ss}ner}.}
  \bibinfo{year}{2019}\natexlab{}.
\newblock \showarticletitle{A scalable verification solution for blockchains}.
\newblock \bibinfo{journal}{\emph{arXiv preprint arXiv:1908.04756}}
  (\bibinfo{year}{2019}).
\newblock


\bibitem[\protect\citeauthoryear{Tramer and Boneh}{Tramer and Boneh}{2018}]%
        {slalom}
\bibfield{author}{\bibinfo{person}{Florian Tramer} {and} \bibinfo{person}{Dan
  Boneh}.} \bibinfo{year}{2018}\natexlab{}.
\newblock \showarticletitle{Slalom: Fast, Verifiable and Private Execution of
  Neural Networks in Trusted Hardware}. In
  \bibinfo{booktitle}{\emph{International Conference on Learning
  Representations}}.
\newblock


\bibitem[\protect\citeauthoryear{Tsankov, Dan, Drachsler-Cohen, Gervais,
  Buenzli, and Vechev}{Tsankov et~al\mbox{.}}{2018}]%
        {securify}
\bibfield{author}{\bibinfo{person}{Petar Tsankov}, \bibinfo{person}{Andrei
  Dan}, \bibinfo{person}{Dana Drachsler-Cohen}, \bibinfo{person}{Arthur
  Gervais}, \bibinfo{person}{Florian Buenzli}, {and} \bibinfo{person}{Martin
  Vechev}.} \bibinfo{year}{2018}\natexlab{}.
\newblock \showarticletitle{Securify: Practical security analysis of smart
  contracts}. In \bibinfo{booktitle}{\emph{Proceedings of the 2018 ACM SIGSAC
  Conference on Computer and Communications Security}}.
  \bibinfo{pages}{67--82}.
\newblock


\bibitem[\protect\citeauthoryear{van~den Hooff, Kaashoek, and
  Zeldovich}{van~den Hooff et~al\mbox{.}}{2014}]%
        {versum}
\bibfield{author}{\bibinfo{person}{Jelle van~den Hooff},
  \bibinfo{person}{M~Frans Kaashoek}, {and} \bibinfo{person}{Nickolai
  Zeldovich}.} \bibinfo{year}{2014}\natexlab{}.
\newblock \showarticletitle{Versum: Verifiable computations over large public
  logs}. In \bibinfo{booktitle}{\emph{Proceedings of the 2014 ACM SIGSAC
  Conference on Computer and Communications Security}}.
  \bibinfo{pages}{1304--1316}.
\newblock


\bibitem[\protect\citeauthoryear{Victor and L{\"u}ders}{Victor and
  L{\"u}ders}{2019}]%
        {erc20}
\bibfield{author}{\bibinfo{person}{Friedhelm Victor} {and}
  \bibinfo{person}{Bianca~Katharina L{\"u}ders}.}
  \bibinfo{year}{2019}\natexlab{}.
\newblock \showarticletitle{Measuring ethereum-based erc20 token networks}. In
  \bibinfo{booktitle}{\emph{International Conference on Financial Cryptography
  and Data Security}}. Springer, \bibinfo{pages}{113--129}.
\newblock


\bibitem[\protect\citeauthoryear{Villa, Chavarria-Miranda, Gurumoorthi,
  M{\'a}rquez, and Krishnamoorthy}{Villa et~al\mbox{.}}{2009}]%
        {villa2009effects}
\bibfield{author}{\bibinfo{person}{Oreste Villa}, \bibinfo{person}{Daniel
  Chavarria-Miranda}, \bibinfo{person}{Vidhya Gurumoorthi},
  \bibinfo{person}{Andr{\'e}s M{\'a}rquez}, {and} \bibinfo{person}{Sriram
  Krishnamoorthy}.} \bibinfo{year}{2009}\natexlab{}.
\newblock \showarticletitle{Effects of floating-point non-associativity on
  numerical computations on massively multithreaded systems}. In
  \bibinfo{booktitle}{\emph{Proceedings of Cray User Group Meeting (CUG)}}.
  \bibinfo{pages}{3}.
\newblock


\bibitem[\protect\citeauthoryear{W{\"u}st, Matetic, Egli, Kostiainen, and
  Capkun}{W{\"u}st et~al\mbox{.}}{2020}]%
        {ace}
\bibfield{author}{\bibinfo{person}{Karl W{\"u}st}, \bibinfo{person}{Sinisa
  Matetic}, \bibinfo{person}{Silvan Egli}, \bibinfo{person}{Kari Kostiainen},
  {and} \bibinfo{person}{Srdjan Capkun}.} \bibinfo{year}{2020}\natexlab{}.
\newblock \showarticletitle{ACE: Asynchronous and Concurrent Execution of
  Complex Smart Contracts}. In \bibinfo{booktitle}{\emph{Proceedings of the
  2020 ACM SIGSAC Conference on Computer and Communications Security}}.
  \bibinfo{pages}{587--600}.
\newblock


\bibitem[\protect\citeauthoryear{Zhou, Qin, Torres, Le, and Gervais}{Zhou
  et~al\mbox{.}}{2020}]%
        {trader}
\bibfield{author}{\bibinfo{person}{Liyi Zhou}, \bibinfo{person}{Kaihua Qin},
  \bibinfo{person}{Christof~Ferreira Torres}, \bibinfo{person}{Duc~V Le}, {and}
  \bibinfo{person}{Arthur Gervais}.} \bibinfo{year}{2020}\natexlab{}.
\newblock \showarticletitle{High-Frequency Trading on Decentralized On-Chain
  Exchanges}.
\newblock \bibinfo{journal}{\emph{arXiv preprint arXiv:2009.14021}}
  (\bibinfo{year}{2020}).
\newblock


\bibitem[\protect\citeauthoryear{Zhu, Ding, and Huang}{Zhu
  et~al\mbox{.}}{2019}]%
        {2pc}
\bibfield{author}{\bibinfo{person}{Ruiyu Zhu}, \bibinfo{person}{Changchang
  Ding}, {and} \bibinfo{person}{Yan Huang}.} \bibinfo{year}{2019}\natexlab{}.
\newblock \showarticletitle{Efficient publicly verifiable 2pc over a blockchain
  with applications to financially-secure computations}. In
  \bibinfo{booktitle}{\emph{Proceedings of the 2019 ACM SIGSAC Conference on
  Computer and Communications Security}}. \bibinfo{pages}{633--650}.
\newblock


\end{thebibliography}
